\RequirePackage{fix-cm}
\documentclass[smallextended]{svjour3}       
\smartqed  
\usepackage{amsmath,amsfonts,amssymb,booktabs,siunitx,bm,xfrac,nicefrac}
\usepackage{algorithm,algpseudocode}
\usepackage{xcolor}
\usepackage{graphicx,subcaption}
\usepackage[hidelinks]{hyperref}
\hypersetup{
    colorlinks,
    linkcolor={green!50!black},
    citecolor={blue!50!black},
    urlcolor={blue!80!black}
}  
\usepackage{relsize}
\begin{document}

\title{A high-order augmented basis positivity-preserving discontinuous Galerkin method for a Linear Hyperbolic Equation}

\titlerunning{Augmented DG method for linear hyperbolic equation}        

\author{
	Maurice S. Fabien
}


\institute{Maurice S. Fabien \at
              Department of Mathematics, University of Wisconsin, Madison, WI 53706, USA \\
              \email{mfabien@wisc.edu}           
}

\date{Received: date / Accepted: date}

\maketitle

\begin{abstract}
This paper designs a high-order positivity-preserving discontinuous Galerkin (DG) scheme for a linear hyperbolic equation.  The scheme relies on augmenting the standard polynomial DG spaces with additional basis functions.  The purpose of these augmented basis functions is to ensure the preservation of a positive cell average for the unmodulated DG solution.  As such, the simple Zhang and Shu limiter~\cite{zhang2010maximum} can be applied with no loss of accuracy for smooth solutions, and the cell average remains unaltered. A key feature of the proposed scheme is its implicit generation of suitable augmented basis functions. Nonlinear optimization facilitates the design of these augmented basis functions. Several benchmarks and computational studies demonstrate that the method works well in two and three dimensions. 
\subclass{MSC code1 \and MSC code2 \and more}
\end{abstract}

\section{Introduction and literature review}
The main objective of this paper is to design a high-order positivity-preserving discontinuous Galerkin (DG) scheme for a linear hyperbolic equation (see equation \eqref{eq:model}). The main idea is to enrich the standard finite element spaces with additional functions. The main purpose for these functions is to ensure that the resulting cell average from the augmented DG scheme preserves a positive cell average.

In~\cite{ling2018conservative}, the authors developed a DG method for the same linear hyperbolic problem (with positive inflow conditions and source functions), which relied on augmenting the standard polynomial spaces $\mathcal{P}_1$ (the space of polynomials of total degree less than or equal to one) and $\mathcal{Q}_1$ (the space of polynomials of degree less than one in each space direction) with additional basis functions.  By doing so, they demonstrated that the resulting scheme preserves a positive cell average for the unmodulated DG solution. Here, unmodulated DG solution refers to a DG solution is not modified in any way through post-processing (such as slope limiting, flux correct transport, or simple `fix-up'/'clipping', see \cite{yee2019quadratic}). The present work is also an extension of \cite{fabien2024positivity}, where similar ideas are explored. In particular, explicit non-polynomial augmented basis functions are defined based on the method of characteristics. Then, these augmented basis functions are used to enrich the standard polynomial spaces $\mathcal{P}_k$ and $\mathcal{Q}_k$. These basis functions are specifically constructed to ensure that the resulting scheme preserves a positive cell average for the DG solution without limiters. As such, the straightforward scaling limiter from~\cite{zhang2010maximum} can be applied without any loss of accuracy for smooth solutions, and the cell average remains unaltered.
 
Bound-preserving schemes have received considerable attention in recent years. Various approaches have been developed to address this challenge. An explicit positivity-preserving scheme that utilizes Lax-Wendroff time stepping was developed in \cite{rossmanith2011positivity}.

Operator splitting methods, particularly Strang splitting, have proven effective in the context of bound-preserving schemes. In \cite{rossmanith2011positivity}, the authors combine Strang operator splitting with a discontinuous Galerkin spatial discretization to achieve high-order accuracy in both space and time for the Vlasov-Possion equations. Similarly, Zhang et al.\cite{zhang2016operator} applied operator splitting with a discontinuous Galerkin discretization to a chemotaxis PDE model - by exploiting the splitting, they are able to retain bound-preservation. These operator splitting approaches are particularly valuable as they can achieve temporal accuracy beyond second order.

For parabolic and convection-diffusion systems, several methods have been proposed in \cite{srinivasan2018positivity,guo2015positivity,wen2016application}. Additionally, techniques that combine $hp$-aptive refinement, and composing different limiters has demonstrated to be effective for nonlinear problems on unstructured meshes \cite{kontzialis2013high}. 
 
	Limiters have also been designed and applied to a variety of other linear and nonlinear problems. For instance, Vlasov-Boltzmann transport equations \cite{cheng2012positivity}, extended magnetohydrodynamics equations \cite{zhao2014positivity}, nonlinear shallow water equations   \cite{li2017positivity}, Gaussian closure equations \cite{meena2017positivity}, and radiative transfer equations \cite{zhang2019high}, and compressible Navier--Stokes equations \cite{zhang2019high}.
   
   There also exist a number of works which utilize various optimization techniques to enforce bound constraints. For semi-implicit and operator splitting schemes, see \cite{liu2024optimizationbased,liu2024optimizationbased3,LIU2024113440}. An efficient active set semismooth Newton method was developed for the fully time-implicit DG discretization of nonlinear problems in \cite{van2019positivity}. In general, provable bound-preserving schemes for fully time-implicit DG methods still remains a challenge - even for linear PDEs.
   
 We note that this list of references on limiters is far from complete. It is quite challenging to prove that the resulting schemes retain high-order accuracy, positivity, as well as preserve other desirable quantities (such as: momentum, energy, and so on). Time-dependent problems introduce additional challenges if implicit time-stepping is considered. In practice, implicit time-stepping gives rise to nonlinear systems that need to be resolved via iterations such as Newtons method. Imposing additional constraints on the nonlinear system can cause issues with convergence \cite{nocedal2006numerical}.
  
In this paper, we extend the results in~\cite{fabien2024positivity,ling2018conservative} by considering the case of polynomial degree $k>1$, and also develop a new strategy related to the work in \cite{fabien2024positivity}.
The work done in \cite{ling2018conservative} produces an explicitly formed defined augmented basis function, which when augmented to the standard DG approximation spaces, results in a positive cell average for the unmodulated DG solution. However, the results are only valid for $k=1$. In~\cite{fabien2024positivity}, an explicitly defined augmented basis function is constructed. Here, specific non-polynomial augmented basis functions are designed using the method of characteristics.   However, this work focuses on producing appropriate augmented basis functions on the fly via an optimization procedure. The advantage of this is that the augmented basis functions are guaranteed to be polynomials.  In particular, we additionally augment the standard DG spaces $\mathcal{P}_k$ (the space of polynomials of total degree less than or equal to $k$) and $\mathcal{Q}_k$ (the space of polynomials of degree less than $k$ in each space direction) with additional basis functions so that the resulting scheme is high-order accurate and preserves a positive cell average for the unmodulated DG solution. The so-called serendipity elements (see \cite{arnold2002approximation}) are also examined. 
  Our approach utilizes a numerical optimization procedure to generate suitable augmented basis functions 
automatically. This is done because it is difficult to a priori explicitly exhibit suitable augmented basis functions for $k>1$.  Some of the notable features of this approach are:
\begin{itemize}
	\item  The unmodulated DG solution for any $k>0$ will have a positive cell average for any Courant–Friedrichs–Lewy (CFL) number~\cite{courant1967partial,leveque1992numerical}.
\item  The resulting method can be used for non-negative variable coefficients.
\item  The algorithm is independent of the domain dimension, therefore, 2D and 3D meshes can be used.
\end{itemize} 
We mention some other related works. It is emphasized that the following three works do not consider function space augmentation. Optimization techniques have been used to enforce that the DG approximation is non-negative everywhere, for instance \cite{van2019positivity,yee2019quadratic}.  Here, a constrained optimization problem is solved where the entire DG solution is constrained to be bound-preserving. The difference between the current work is that we merely request that the cell average remains positive, and not the entire DG approximation. This reduced the computational effort, and relaxes the complexity of the underlying optimization problem.  Similarly, the works in \cite{liu2024optimizationbased,liu2024optimizationbased3,LIU2024113440} also leverage optimization procedures to obtain bound-preserving schemes, but they do not consider function space augmentation, which is the primary focus of the current work.

	In \cite{XU2022111410} the authors consider under-integration (inexact quadrature), using quadrature rules with reduced precision. By doing this, they show good results and some theoretical evidence that under-integration can produce positive cell averages for the unmodulated unaugmented DG solver. Of course, the reduced precision of the quadrature effects the overall accuracy for smooth solutions.
 
	The work in \cite{XU2023112304} explores new notions of conservation which are arguably more suitable for steady-state or implicit problems. It is found that the standard notion of conservation, preserving the cell average, may not be ideal. For instance, this may result in mass loss over long time scales. By redefining conservation, they are able to design high-order positivity-preserving DG schemes without the need for augmentation. The limitation is of course that the classical notion of conservation must be abandoned; and it is not immediately clear how to extend the results to nonlinear problems.


The current paper is organized as follows. Sections \ref{section:model} and \ref{section:model3D} describe the model problems and briefly introduce the DG discretization. The strategy for obtaining suitable augmented basis functions is introduced in Section \ref{section:procedure}. To ensure the overall DG approximation remains non-negative, the simple scaling limiter from~\cite{zhang2010maximum} is rehashed in Section \ref{section_limiter}. Section \ref{section_motivations} presents several examples motivating function space augmentation, as well as explicit basis functions. Numerical experiments are conducted in Section \ref{section:numerics}, which verify and validate the numerical optimization strategy. The paper is finalized with important findings and conclusions.
\section{Model problem and discretization in (1+1) dimensions}\label{section:model}
The two dimensional linear hyperbolic equation considered in this paper has the form
\begin{equation} 
  \frac{   \partial (\alpha u (x,y))}{\partial x}
+
  \frac{   \partial (\beta u(x,y)) }{\partial y}
+
\gamma u(x,y)
=
f(x,y),
\label{eq:model}
\end{equation}
where $\alpha,\beta>0, \,\gamma\ge 0,\, f(x,y)\ge0$, and the domain is $\mathcal{D} = [a,b]\times [c,d] \subset \mathbb{R}^2.$ If equation~\eqref{eq:model} is prescribed non-negative initial/inflow boundary conditions, then one can show that the solution $u$ will be non-negative (e.g., method of characteristics).  The computational domain is partitioned into a rectangular mesh with cells $S_{ij} = [x_{\scriptscriptstyle{i- \frac{1}{2}}},x_{\scriptscriptstyle{i+ \frac{1}{2}}}]\times [y_{\scriptscriptstyle{j- \frac{1}{2}}},y_{\scriptscriptstyle{j+ \frac{1}{2}}}]$, with $i=1,\ldots,N_x$ and $j=1,\ldots, N_y$.  Define the traditional DG space $V_h^k$ as
\[
V_h^k = \{v\in L^2(\mathcal{D}): v|_{S_{ij}} \in \mathcal{X}_k(S_{ij}),\,\forall i=1,\ldots,N_x,\, \forall j = 1,\ldots, N_y\},
\]
where $\mathcal{X}_k$ is a standard polynomial space such as $\mathcal{P}_k$ (space of all polynomials up to degree $k$), $\mathcal{Q}_k$ (space of tensor-product polynomials of degree $k$), or the serendipity space $\mathcal{S}_k$ (for details, see \cite{arnold2002approximation}) with polynomial degree $k\ge 1$.  The DG scheme for solving~\eqref{eq:model} then seeks $u \in V_h^k$ so that for all $v\in V_h^k$ we have
\begin{equation} 
\mathcal{L}(u,v) = \mathcal{R}(v),
\label{eq_origin}
\end{equation}
where the bilinear form $\mathcal{L}(\cdot,\cdot)$ and functional $\mathcal{R}(\cdot)$ are given by
\begin{subequations}
\begin{align} 
\mathcal{L}(u,v) &= -\int_{S_{ij}} u(x,y) (\alpha v_x(x,y) + \beta v_y(x,y))  \,dx \, dy
\\
&+
  \int_{y_{\scriptscriptstyle{j- \frac{1}{2}}}}^{y_{\scriptscriptstyle{j+ \frac{1}{2}}}} \alpha u(x^-_{ \scriptscriptstyle{i+ \frac{1}{2}}},y)  v(x^-_{i+\scriptscriptstyle\frac{1}{2}},y)     \,dy
\notag 
\\
&+
   \int_{x_{\scriptscriptstyle{i- \frac{1}{2}}}}^{x_{\scriptscriptstyle{i+ \frac{1}{2}}}} \beta u(x,y^-_{j+\scriptscriptstyle\frac{1}{2}})  v(x,y^-_{\scriptscriptstyle{j+\frac{1}{2}}})    \,dx
+
  \int_{S_{ij}} \gamma u(x,y) v(x,y) \,dx \,dy,
\notag 
\\
\mathcal{R}(v)&= \int_{S_{ij}} f(x,y) v(x,y) \,dx \,dy
+ 
  \int_{y_{\scriptscriptstyle{j- \frac{1}{2}}}}^{y_{\scriptscriptstyle{j+ \frac{1}{2}}}}
\alpha u(x^-_{ \scriptscriptstyle{i- \frac{1}{2}}},y) v(x^+_{ \scriptscriptstyle{i- \frac{1}{2}}},y)
\,dy
\\
&+
   \int_{x_{\scriptscriptstyle{j- \frac{1}{2}}}}^{x_{\scriptscriptstyle{j+ \frac{1}{2}}}}
 \beta u(x,y^-_{j-\scriptscriptstyle\frac{1}{2}})  v(x,y^+_{\scriptscriptstyle{j-\frac{1}{2}}})    \,dx  .
\notag 
\end{align}
\label{eq:dg_scheme}
\end{subequations}
Equation~\eqref{eq_origin} can be solved locally on each cell, with inflow boundary conditions.  In other words, inflow conditions on a given cell are outflow conditions from an adjacent cell or inflow from the domain boundary.

The cell average of a function $u$ on cell $S_{ij}$ is denoted by
\[
\overline{u}_{ij}
= 
\frac{1}{|S_{ij}|}\int_{S_{ij}} u \,dx\,dy,
\]
where $|S_{ij}|$ is the measure (area) of cell $S_{ij}$.

\section{Model problem and discretization in (1+2) dimensions}\label{section:model3D}
Consider the (1+2) dimension problem:
\begin{equation} 
\frac{\partial u (x,y,t)}{\partial t} + \alpha\frac{\partial  u (x,y,t)}{\partial x} + \beta \frac{\partial u (x,y,t)}{\partial y} + \gamma u(x,y,t) = f(x,y,t), 
\label{eq:model3D}
\end{equation}
with the space-time domain $(x,y,t) \in \mathcal{D} =  \Omega \times [0,T]$ (for some final time $T>0$), $\Omega\subset \mathbb{R}^2$.  Similar to the model problem in Section~\ref{section:model}, $\alpha,\beta>0$ and $\gamma\ge0$  are constants.  The PDE~\eqref{eq:model3D} is equipped with non-negative inflow boundary conditions, and it is assumed that $f\ge0$.
 
We employ a space-time DG method, so that the space-time slab is now a rectangular cuboid (instead of a rectangle).  In other words, the space-time domain $\Omega \times [0,T]$ is partitioned into a total of $N_x*N_y*N_t$ rectangular cuboids $S_{ij\ell}$ ($N_x$ cuboids in the $x$-direction, $N_y$ cuboids in the $y$-direction, $N_t$ cuboids in the $t$-direction).  The cuboid mesh then has cells $S_{ij\ell} = [x_{\scriptscriptstyle{i- \frac{1}{2}}},x_{\scriptscriptstyle{i+ \frac{1}{2}}}]\times [y_{\scriptscriptstyle{j- \frac{1}{2}}},y_{\scriptscriptstyle{j+ \frac{1}{2}}}]\times [t_{\scriptscriptstyle{\ell- \frac{1}{2}}},t_{\scriptscriptstyle{\ell+ \frac{1}{2}}}]$, with $i=1,\ldots,N_x$,  $j=1,\ldots, N_y$ and $\ell=1,\ldots, N_t$.

	The DG scheme for solving~\eqref{eq:model3D} then seeks $u \in V_h^k$ so that for all $v\in V_h^k$ we have  
\begin{align} 
&-\int_{S_{ij\ell}} u  (\alpha v_x  + \beta v_y + v_t - \gamma v)  \,dx \, dy\,dt
\label{eq:dg_scheme3D}
\\
&+
\alpha 
\int_{t_{\scriptscriptstyle{j- \frac{1}{2}}}}^{t_{\scriptscriptstyle{j+ \frac{1}{2}}}} 
\int_{y_{\scriptscriptstyle{j- \frac{1}{2}}}}^{y_{\scriptscriptstyle{j+ \frac{1}{2}}}} 
u(x^-_{ \scriptscriptstyle{i+ \frac{1}{2}}},y,t)  v(x^-_{i+\scriptscriptstyle\frac{1}{2}},y,t)     \,dy \,dt
\notag 
\\
&+
\beta  
\int_{t_{\scriptscriptstyle{j- \frac{1}{2}}}}^{t_{\scriptscriptstyle{j+ \frac{1}{2}}}} 
\int_{x_{\scriptscriptstyle{j- \frac{1}{2}}}}^{x_{\scriptscriptstyle{j+ \frac{1}{2}}}} u(x,y^-_{i+\scriptscriptstyle\frac{1}{2}},t)  v(x,y^-_{\scriptscriptstyle{i+\frac{1}{2}}},t)    \,dx\,dt
\notag 
\\
&+
\int_{y_{\scriptscriptstyle{j- \frac{1}{2}}}}^{y_{\scriptscriptstyle{j+ \frac{1}{2}}}} 
\int_{x_{\scriptscriptstyle{j- \frac{1}{2}}}}^{x_{\scriptscriptstyle{j+ \frac{1}{2}}}} u(x,y,t^-_{\scriptscriptstyle{i+ \frac{1}{2}}})  v(x,y,t^-_{\scriptscriptstyle{i+\frac{1}{2}}})    \,dx\,dy
\notag 
\\
&= \int_{S_{ij\ell}} f  v  \,dx \,dy\,dt
+
\alpha 
\int_{t_{\scriptscriptstyle{j- \frac{1}{2}}}}^{t_{\scriptscriptstyle{j+ \frac{1}{2}}}}
\int_{y_{\scriptscriptstyle{j- \frac{1}{2}}}}^{y_{\scriptscriptstyle{j+ \frac{1}{2}}}}
u(x^-_{ \scriptscriptstyle{i- \frac{1}{2}}},y,t) v(x^+_{ \scriptscriptstyle{i- \frac{1}{2}}},y,t)
\,dy\,dt
\notag 
\\
&+ 
\beta 
\int_{t_{\scriptscriptstyle{\ell- \frac{1}{2}}}}^{t_{\scriptscriptstyle{\ell+ \frac{1}{2}}}}
 \int_{x_{\scriptscriptstyle{i- \frac{1}{2}}}}^{x_{\scriptscriptstyle{i+ \frac{1}{2}}}}
 u(x,y^-_{j-\scriptscriptstyle\frac{1}{2}},t)  v(x,y^+_{\scriptscriptstyle{j-\frac{1}{2}}},t)    \,dx  \,dt
 \notag 
\\
&+ 
\int_{y_{\scriptscriptstyle{j- \frac{1}{2}}}}^{y_{\scriptscriptstyle{j+ \frac{1}{2}}}}
 \int_{x_{\scriptscriptstyle{i- \frac{1}{2}}}}^{x_{\scriptscriptstyle{i+ \frac{1}{2}}}}
 u(x,y,t^-_{\ell-\scriptscriptstyle\frac{1}{2}})  v(x,y,t^+_{\scriptscriptstyle{\ell-\frac{1}{2}}})    \,dx  \,dy 
 .
\notag 
\end{align} 

The cell average of a function $u$ on cell $S_{ij\ell}$ is denoted by
\[
\overline{u}_{ij\ell}
= 
\frac{1}{|S_{ij\ell}|}\int_{S_{ij\ell}} u \,dx\,dy \,dt,
\]
where $|S_{ij\ell}|$ is the measure (volume) of cell $S_{ij\ell}$.

\section{Procedure for high-order positivity-preserving scheme}\label{section:procedure}
 It is possible to augment standard DG spaces with additional basis functions so that high-order accuracy is retained and the cell average is guaranteed to be positive.  In~\cite{ling2018conservative}, the authors augment the standard $\mathcal{P}_1$ space by adding additional basis functions which guarantee that the cell average of the unmodulated DG solution remains positive.  
As long as the inflow conditions and source are also non-negative, the cell average of the unmodulated DG solution is guaranteed to be non-negative (which can be deduced from~\eqref{eq:dg_scheme}).
  
The focus of this paper is to extend the results in~\cite{ling2018conservative}.  More specifically, assuming positive inflow and sources, we augment the standard DG spaces ($k>1$) so that the resulting DG scheme is high-order and preserves a positive cell average for the unmodulated DG solution.  This is accomplished by utilizing an optimization algorithm which numerically searches for appropriate basis functions to augment the standard DG spaces.  

Here we search for a \textit{single} basis function to augment the standard DG space with.  The augmented basis function (denoted by $\psi$) is treated as an ansatz (we select a polynomial of degree $r>k$), and search for coefficients of $\psi$ (via nonlinear programming) which ensure that the cell average of the unmodulated DG solution remains positive. 

\subsection{Augmenting polynomial basis}
Here we provide an algorithm that generates a special basis function; such that, when it is augmented to the standard DG spaces, it will preserve a positive cell average for the unmodulated DG solution.  Let $\mathcal{X}_k$ be a traditional polynomial space (e.g., $\mathcal{P}_k$, $\mathcal{Q}_k$, or the serendipity space $\mathcal{S}_k$).  Following~\cite{ling2018conservative}, for $k\ge 1$, we augment $\mathcal{X}_k$ with an additional basis function, denoted by $\psi $.  The polynomial $\psi$ will belong to some higher-order space $\mathcal{X}_r(S_{ij})$ with $r>k$. That is, $\psi \in \mathcal{X}_r(S_{ij})\backslash\mathcal{X}_k(S_{ij})$.  The space $\mathcal{X}_k$ augmented with $\psi$ is denoted by $\mathcal{ \widetilde{X}}_k = \mathcal{X}_k \cup \text{span}(\psi)$.  Let $\widetilde{V}_h^k$ be the augmented DG space:
\[
\widetilde{V}_h^k = \{v\in L^2(\mathcal{D}): v|_{S_{ij}} \in \mathcal{ \widetilde{X}}_k(S_{ij}),\,\forall i=1,\ldots,N_x,\, \forall j = 1,\ldots, N_y\} .
\] 
We search for some higher-order polynomial $\psi$ that will generate a test function $\widetilde{v} \in \widetilde{V}_h^k$ which simultaneously satisfies $ \widetilde{v}  \ge 0$ and
\begin{equation} 
 \mathcal{L}(u,\widetilde{v}) = \int_{S_{ij}}u\,dx\,dy,  
\label{eq:LHS_avg}
\end{equation}
 for any $u\in \widetilde{V}_h^k$.
 
 
Let $\Phi_i$ be the $i$th basis function for $\mathcal{X}_r$, with $1\le i\le \dim{\mathcal{X}_r}$.  In addition, set $\psi = \sum_{i=1}^{ \dim{\mathcal{X}_r} } \vec{d}_i \Phi_i,$ and $\vec{d} =[d_1,d_2,\ldots,d_{\dim{\mathcal{X}_r}}]^T$.  Then, we pose the following optimization problem: find a test function $\psi:= \widetilde{v}\in \widetilde{\mathcal{X}}_k(S_{ij})$ (not necessarily unique) such that
\begin{subequations}
\begin{alignat}{2}
&\!\min_{ \vec{d} }        &\qquad& \|\vec{d}\|^2_2 \label{eq:optProb}\\
&\text{subject to} &      & \widetilde{v}(x,y) \geq 0, ~(x,y)\in S_{ij},\label{eq:constraint1}\\
&                  &      &  \mathcal{L}(\widetilde{u},\widetilde{v})=\int_{S_{ij}}\widetilde{u}\,dx\,dy, \quad \forall \widetilde{u} \in \widetilde{\mathcal{X}}_k(S_{ij}),\label{eq:constraint2}.
\end{alignat}
\label{eq:optimization_algorithm}
\end{subequations}
In full generality, this optimization problem can be numerically resolved on each troubled cell locally by using nonlinear programming techniques such as the interior point method~\cite{byrd2000trust}.
 
	Of course, this is an ansatz, we \textit{assume} that there is such a polynomial $\psi$ that satisfies Problem~\eqref{eq:optimization_algorithm}.  In general, it is difficult to verify that Problem~\eqref{eq:optimization_algorithm} has a solution (e.g., establishing the Karush–Kuhn–Tucker conditions~\cite{nocedal2006numerical}).  Numerically we are able to obtain feasible solutions to Problem~\eqref{eq:optimization_algorithm}.
 
We next describe some steps of Problem~\eqref{eq:optimization_algorithm}.  It should be noted that this procedure can be preprocessed, and applied to a reference cell (say, the square $[-1,1]^2$ in 2D) instead of the physical cell $S_{ij}$.  In practice, the DG scheme in~\eqref{eq_origin} is typically resolved through quadrature; and the approximate solution to is only desired on some subset of $S_{ij}$, for instance quadrature or interpolation points. Let $\mathcal{G} \subset S_{ij}$ be this set of finitely many points (e.g., tensor product Gauss quadrature rules of sufficiently large precision).  Each optimization step of Problem~\eqref{eq:optimization_algorithm} starts with an initial guess for $\vec{d}$ (the coefficients for the candidate augmented basis function $\psi$). Next, the augmented space is formed $\widetilde{\mathcal{X}}_k(S_{ij}) =  {\mathcal{X}}_k(S_{ij}) \cup \,\textrm{span}\{\psi\}$, and the linear system in~\eqref{eq:constraint2} (which has size $1+\dim\mathcal{X}_k$) is assembled and solved for $\widetilde{v}\in \widetilde{\mathcal{X}}_k(S_{ij})$. If $\widetilde{v}$ is nonzero on the desired set of points $\mathcal{G}$, then the optimization step terminates. Otherwise, a new candidate for $\vec{d}$ is generated and the process repeats until $v\ge0$ on the set $G$.  For all numerical experiments and results in this paper, we take $\mathcal{G}$ to be the standard tensor product Gauss quadrature points~\cite{kovvali2011theory}, with enough points to ensure $\int_{S_{ij}} \widetilde{u}\widetilde{v} $ with $\widetilde{u},\widetilde{v}\in \widetilde{\mathcal{X}}_k(S_{ij})$ to be integrated exactly.
 
	Ideas that incorporate optimization in positivity-preserving DG schemes have been considered in~\cite{van2019positivity} and~\cite{yee2019quadratic}.  These works use tools from optimization to design maximum-principle preserving limiters. 
Essentially, they solve a constrained optimization problem where the entire solution is constrained (or modified) to be nonnegative.  In contrast, our approach seeks to maintain a positive cell average.  The optimization procedure in Problem~\eqref{eq:optimization_algorithm} searches for a suitable basis function to augment $\mathcal{X}_k$.  With appropriate augmented basis functions for the standard DG space, we will have that the cell average of the unmodulated DG solution remains positive.  Then, if the unmodulated DG solution is negative, one can apply the scaling limiter described in~\cite{zhang2010maximum}.  As such, the proposed approach has a provably high-order accurate limiter, and, the cell average remains unaltered.

Assuming a troubled cell ($\overline{u}_{ij}<0$), the overall procedure is summarized below:
\begin{enumerate}  
\item  Pick an ansatz polynomial $\psi \in \mathcal{X}_r$ with degree $r>k$, $\psi = \sum_{i=1}^{|\mathcal{X}_r|} \vec{d}_i \Phi_i$.
\item  Numerically solve Problem~\eqref{eq:optimization_algorithm} to determine viable coefficients $\vec{d}$ for $\psi$. 
\item  
	Augment the standard $\mathcal{X}_k$ space with $\psi$ (denote the augmented space by $\mathcal{ \widetilde{X}}_k$); replace $ {V}_h^k$ with $\widetilde{V}_h^k$ in~\eqref{eq:model}.  Then, solve~\eqref{eq:model} for $u_{ij}\in \widetilde{V}_h^k$, where 
\[  
\widetilde{V}_h^k = \{v\in L^2(\mathcal{D}): v|_{S_{ij}} \in \mathcal{ \widetilde{X}}_k(S_{ij}),\,\forall i=1,\ldots,N_x,\, \forall j = 1,\ldots, N_y\} .
\]
\item  If $v<0$, limit the approximation.  Apply the scaling limiter~\cite{zhang2010maximum}, $\tilde{u}_{ij}(x,y) = \theta(u_{ij}(x,y)-\overline{u}_{ij}) + \overline{u}_{ij}$.
\end{enumerate} 
 
Through various numerical experiments, it is found that $r$ is not substantially larger than $k$.  If $\psi$ is be to valid for only one cell $S_{ij}$, then $r=k+1$ appears to suffice.  In the case that it is desirable for a given $\psi$ to work for multiple cells (all cells obeying some Courant–Friedrichs–Lewy (CFL) upper/lower bound), then larger values of $r$ seem to be necessary. This paper utilizes Problem~\eqref{eq:optimization_algorithm} as stated, but there are many ways to modify the search for suitable augmented functions $\psi$.  For instance, one could use non-negative basis functions and enforce $\vec{d}_i\ge 0$.  Other possibilities include (but are not limited to) modifying the objective function (e.g., log barrier/penalty methods for remote initial guesses), augmenting multiple basis functions to a standard DG space, or requesting $\psi$ posses certain properties (e.g., $\psi$ vanishes at the outflow boundaries).

This approach may seem prohibitively expensive, however, Problem~\ref{eq:optimization_algorithm} only needs to be invoked on troubled cells (where the unmodulated cell average is negative).  In other words, we use an adaptive scheme, where the standard $\mathcal{X}_k$ space is used when the resulting cell average is non-negative, and Problem~\ref{eq:optimization_algorithm} is only activated when the computed cell average from $\mathcal{X}_k$ is negative.  Typically this only occurs on a subset of the mesh.  It should be also noted that~\eqref{eq:dg_scheme} is usually approximated using quadrature rules.  As such, we may weaken the constraints of Problem~\eqref{eq:optimization_algorithm}, so that $v$ is only non-negative at certain quadrature or control points on a cell (instead of non-negative over the entire cell).  Moreover, due to the local nature of the DG scheme, Problem~\ref{eq:optimization_algorithm} can be leveraged on a single cell. We also demonstrate that it is possible to pre-process suitable $\psi$ that is valid for all cells obeying a CFL-type bound. Furthermore, Problem~\ref{eq:optimization_algorithm} is extendable to higher dimensional Cartesian meshes. 

\subsection{Positivity-preserving DG scheme in (1+2) dimensions}
To arrive at a positivity-preserving DG scheme for~\eqref{eq:dg_scheme3D}, we use the approach from Section~\ref{section:procedure}.  Let $\mathcal{L}(u,v)$ denote the bilinear form associated with the left hand side of~\eqref{eq:dg_scheme3D}.  The classical DG space is given by
\[
 {V}_h^k = \{v\in L^2(\mathcal{D}): v|_{S_{ij\ell}} \in \mathcal{ {X}}_k(S_{ij\ell}),\, 1\le i\le N_x,\, 1\le j\le N_y, \, 1\le \ell \le N_t\} .
\] 
We seek some higher-order polynomial $\psi$ to augment to the space $V_h^k$ with.  The resulting augmented space $\widetilde{V}_h^k$ will contain a test function $v\in \widetilde{V}_h^k$ that simultaneously satisfies $ \widetilde{v}  \ge 0$ and
\begin{equation} 
 \mathcal{L}(\widetilde{u},\widetilde{v}) = \int_{S_{ij\ell}}\widetilde{u}\,dx\,dy\,dt,  
 \label{eq:LHS_avg3D}
\end{equation}
 for any $\widetilde{u}\in \widetilde{V}_h^k$.  One can obtain an appropriate $\psi \in \mathcal{X}_r$ from Problem~\eqref{eq:optimization_algorithm} applied to $(1+2)$ dimensions: For $r>k$, find a test function $\psi := \widetilde{v}\in \widetilde{\mathcal{X}}_k(S_{ij\ell})$ (not necessarily unique) such that
\begin{subequations}
\begin{alignat}{2}
&\!\min_{ \vec{d}}        &\qquad& \|\vec{d}\|_2^2 \\
&\text{subject to} &      & \widetilde{v}(x,y,t) \geq 0,~ (x,y,t)\in S_{ij\ell},\\
&                  &      &  \mathcal{L}(\widetilde{u},\widetilde{v})=\int_{S_{ij\ell}}\widetilde{u}\,dx\,dy, \quad \forall \widetilde{u} \in \widetilde{\mathcal{X}}_k(S_{ij\ell}).
\end{alignat}
\label{eq:optimization_algorithm3D}
\end{subequations}
Problem~\eqref{eq:optimization_algorithm3D} only needs to be applied on troubled cells (where the cell average from the unmodulated DG scheme using $\mathcal{ {X}}_k(S_{ij\ell})$ is negative).  In practice, Problem~\eqref{eq:optimization_algorithm3D} can be weakened to hold at quadrature points on $S_{ij\ell}$ (and not necessarily over the entire cell).


\remark Problems~\eqref{eq:optimization_algorithm} and ~\eqref{eq:optimization_algorithm3D} searches for coefficients $\vec{d}$ such that the solution $\widetilde{v}\ge 0$ and satisfies $\mathcal{L}(\widetilde{u},\widetilde{v})=\overline{u}$ for all $\widetilde{u}\in \widetilde{\mathcal{X}}_k$. In general, determining these coefficients is not trivial - especially if nonnegativity is requested to hold for \textit{all} admissible parameters $\alpha,\beta,\omega,\gamma,\Delta x, \Delta y,\Delta z$. That is, $\vec{d}$ are independent of the parameter space $\Omega_\delta = (\alpha,\beta,\omega,\gamma,\Delta x, \Delta y,\Delta z)$.

	For a fixed set of parameters $\Omega_\delta$, Problems~\eqref{eq:optimization_algorithm} and ~\eqref{eq:optimization_algorithm3D} is able to generate feasible candidates for $\vec{d}$. However, determining $\vec{d}$ that are feasible for any parameters in $\Omega_\delta$ is more challenging. Effectively, this amounts to multi-parametric (MP) programming problems~\cite{pistikopoulos2020multi} of the type
	\begin{align*}
	J^*(\vec{\theta}) &= \min_{\vec{x}\in \mathbb{R}^n} f(\vec{x},\vec{\theta})
	\\
	&\text{subject to}~~g(\vec{x},\vec{\theta}),~~~\vec{\theta} \in \Omega_\delta, 
	\end{align*}
	where the optimal value $J$ in general \textit{depends} on the parameters $\vec{\theta}$. Obtaining $\vec{d}$ that are constant for all $\vec{\theta}\in \Omega_\delta$ can also be interpreted as a semi-infinite program (SIP) with finitely many variables and an infinite number of constraints~\cite{reemtsen2013semi}. Nonlinear variants of these optimization problems remain a highly active area of research. It is worth pointing out that these SIP and MP optimization techniques assume that $\Omega_\delta$ is a compact set, which automatically rules out non-compact sets such as $\{(\alpha,\beta\,\gamma): \alpha>0,\beta>0, \gamma\ge 0,\}$.
	
	From another perspective, the primary focus is to satisfy the constraints (nonnegativity), as such, the parameter search could also be interpreted as a feasibility problem (no objective function, only constraints).

\section{Positivity-preserving limiter}  \label{section_limiter} 
In this section we summarize the positivity-preserving scheme.  The augmented spaces $\widetilde{\mathcal{X}}_k$ only ensure that the cell average from the unmodulated DG scheme is non-negative.  To arrive at a solution that is non-negative everywhere, a limiter may be needed.  We first state the following important conclusion, then discuss the limiter.
\begin{theorem} 
Consider the model PDE~\eqref{eq:model} or~\eqref{eq:model3D}.  Then, the cell average $\overline{u}|_S$ of the DG schemes in~\eqref{eq:dg_scheme} and~\eqref{eq:dg_scheme3D} with the augmented spaces $\widetilde{ \mathcal{X}}_k(S)$ will remain non-negative provided that the source term and inflow conditions from the upstream cells (including the inflow from the boundary conditions) are non-negative.
\end{theorem}
\begin{proof} 
 Let $S\in \mathcal{T}_h$ and set the polynomail degree $k\ge 1$. Assume that the inflow conditions and source function $f$ are non-negative. Further, suppose that Problem~\eqref{eq:optimization_algorithm} has generated an augmented basis function $\psi \in \mathcal{X}_r(S)$, for $r>k$ (see Appendix A for special cases). As such, there exists $ \widetilde{v}\in \widetilde{X}_k(S) = \mathcal{X}_k(S) \cup \text{span}\{\psi\}$ that simultaneously satisfies $0\le \widetilde{v}$ and
\begin{align}
\mathcal{L}(\widetilde{u},\widetilde{v}) = \int_S \widetilde{u} \,dx \,dy  ,
\label{eq_proof1}
\end{align} 
for any $\widetilde{u}\in \widetilde{X}_k(S)$ (where $\mathcal{L}(\cdot,\cdot)$ is defined in~\eqref{eq:dg_scheme}). However, from \eqref{eq_origin}, we also have
\begin{align}
\mathcal{R}(\widetilde{v}) = \mathcal{L}(\widetilde{u},\widetilde{v}),
\label{eq_proof2}
\end{align} 
where $\mathcal{R}(\widetilde{v}) \ge0 $, since we assumed $f $ and the inflow conditions are both non-negative. Using \eqref{eq_proof2} in \eqref{eq_proof1} gives
\[
0\le \mathcal{R}(\widetilde{v}) = \mathcal{L}(\widetilde{u},\widetilde{v})= \int_S \widetilde{u} \,dx \,dy.
\] 
It follows that $\overline{u}_S = \frac{1}{|S|} \int_S \widetilde{u} \,dx \,dy\ge 0. $
\hfill 
$\blacksquare$
\end{proof}

The limiter we use is applied at quadrature points.  In practice, equation~\eqref{eq:LHS_avg} is approximated by quadrature.  This suggests we may weaken the condition which enforces that the test function $v$ from~\eqref{eq:optimization_algorithm} is non-negative on an entire cell $S_{ij}$.  Instead, one may request that $v$ from~\eqref{eq:optimization_algorithm} is non-negative only at specific quadrature points (and not necessarily everywhere on $S_{ij}$).  If the corresponding quadrature weights are positive, it then follows that the resulting quadrature will be non-negative (similar conclusions hold for the DG scheme in $(1+2)$ dimensions).

 Let the set $G_{2D}$ contain the abscissa on cell $S_{ij}$ from a tensor product quadrature rule with non-negative weights (such as the Gauss-Lobatto quadrature~\cite{abramowitz1965handbook,quarteroni2001numerical}).  Similarly, in 3D, the set $G_{3D}$ contains the abscissa on cell $S_{ij\ell}$ from a tensor product quadrature rule with non-negative weights.  In 2D, the scaling limiter from~\cite{zhang2010maximum} modifies the approximation $u_{ij}(x,y)$ such that 
\[
\widetilde{u}_{ij}(x,y) = \theta (u_{ij}(x,y) - \overline{u}_{ij} ) +  \overline{u}_{ij},
\quad
\theta = \min_{ (x,y) \in G_{2D} } \bigg\{ 1, \bigg| \frac{\overline{u}_{ij}}{u_{ij}(x,y) - \overline{u}_{ij}}  \bigg| \bigg\}.
\]
 The 3D analog of the scaling limiter is
  \[
\widetilde{u}_{ij \ell}(\vec{x}) = \theta (u_{ij \ell}(\vec{x}) - \overline{u}_{ij \ell} ) +  \overline{u}_{ij \ell},
\quad
\theta = \min_{ \vec{x} \in G_{3D} } \bigg\{ 1, \bigg| \frac{\overline{u}_{ij \ell}}{u_{ij \ell}(\vec{x}) - \overline{u}_{ij \ell}}  \bigg| \bigg\},
\]   
where $\vec{x}=[x,y,z]^T$.
The limited polynomial solution $\widetilde{u} $ is non-negative and also retains high-order accuracy~\cite{zhang2010maximum}.  
   
\section{Motivating examples for problem~\eqref{eq:optimization_algorithm}}   
\label{section_motivations}
The purpose of this section is to showcase some features of Problem~\eqref{eq:optimization_algorithm}.  In particular, the effectiveness of Problem~\eqref{eq:optimization_algorithm} is demonstrated in the case of variable coefficients and polynomial degrees $k>1$.  Several motivating examples are also provided which illustrate the need for augmentation.  This is done by examining the traditional unaugmented high-order DG schemes (e.g., with $\mathcal{P}_k$, $\mathcal{S}_k$, or $\mathcal{Q}_k$); to see whether or not they give rise to a test function $v$ that satisfies both $v\ge0$ and~\eqref{eq:LHS_avg}.  The condition $v\ge0$ is enforced because we need the right hand side of~\eqref{eq:dg_scheme} to be non-negative to guarantee a non-negative cell average.  In more detail, we enforce that~\eqref{eq:LHS_avg} holds, however, \eqref{eq:dg_scheme} gives
 \begin{align*}
\int_{S_{ij}} u \,dx\,dy &=    \mathcal{L}(u,v)
 \\
 &= \int_{S_{ij}} f v  \,dx \,dy
+ 
\alpha \int_{y_{\scriptscriptstyle{j- \frac{1}{2}}}}^{y_{\scriptscriptstyle{j+ \frac{1}{2}}}}
u(x^-_{ \scriptscriptstyle{i- \frac{1}{2}}},y) v(x^+_{ \scriptscriptstyle{i- \frac{1}{2}}},y)
\,dy
\\
&+ 
\beta  \int_{x_{\scriptscriptstyle{j- \frac{1}{2}}}}^{x_{\scriptscriptstyle{j+ \frac{1}{2}}}}
 u(x,y^-_{i-\scriptscriptstyle\frac{1}{2}})  v(x,y^+_{\scriptscriptstyle{i-\frac{1}{2}}})    \,dx .
 \end{align*}
 Since we assume that $f\ge0$ and the inflow from $u$ is also non-negative, it follows that $v\ge0$ ensures the cell average of $u$ will be non-negative.
\subsection{The need for augmented spaces in (1+1) dimensions}
  We consider a single cell $S_{ij}=[-1,1]^2$, and $\alpha=1$, $\beta=1/10$, $\gamma=0$, $\Delta x_i=\Delta y_j = 0.1$. By solving~\eqref{eq:LHS_avg} (replacing $ {\widetilde{V}}_k $ with $ { {V}}_k $) with $r=k=2$, one finds that the resulting unique test function $v$ has a negative minimum (see Fig.~\ref{fig:Xk_motive}).
\begin{figure}[htb!]
\centering
\begin{subfigure}[h]{0.4\linewidth}
\includegraphics[width=\linewidth]{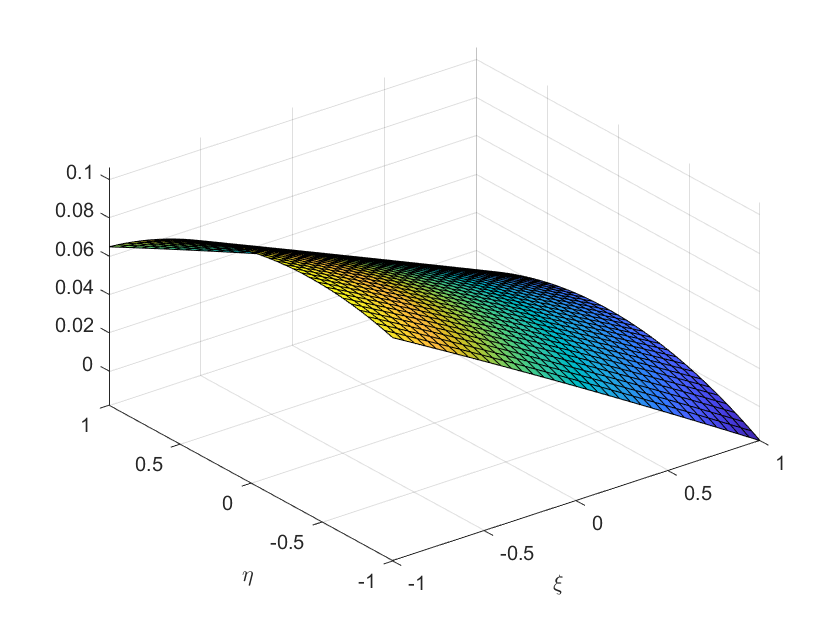}
\caption{{\small $v$ region plot for $\mathcal{P}_2$}}
\label{fig:P2_motive1}
\end{subfigure}
\hspace*{5ex}
\begin{subfigure}[h]{0.4\linewidth}
\includegraphics[width=\linewidth]{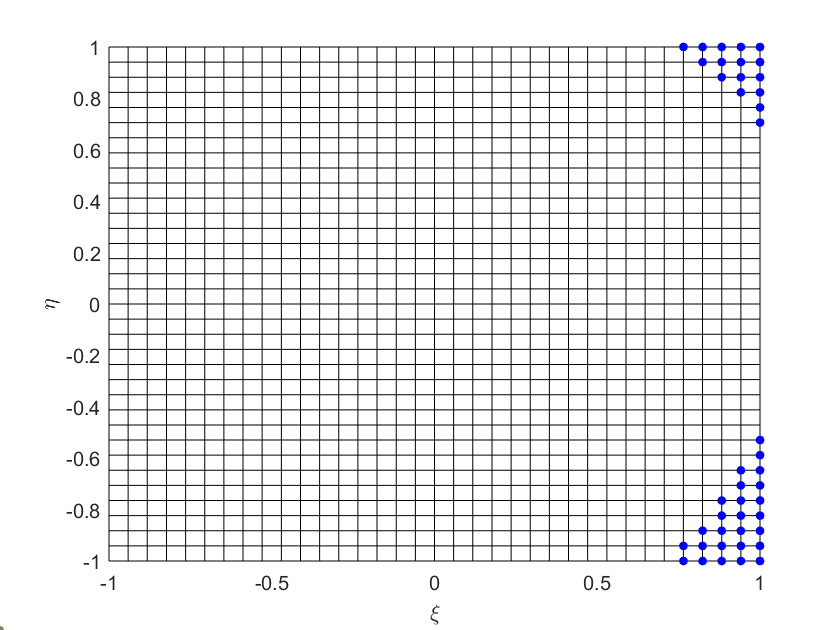}
\caption{{\small Dots are negative values ($\mathcal{P}_2$)}}
\label{fig:P2_motive2}
\end{subfigure}%
\\
\begin{subfigure}[h]{0.4\linewidth}
\includegraphics[width=\linewidth]{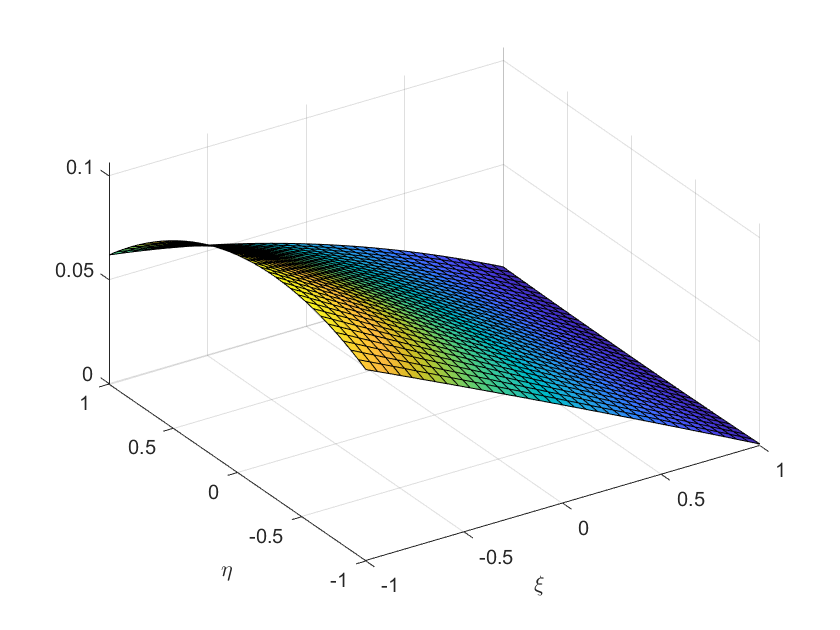}
\caption{{\small $v$ region plot for $\mathcal{Q}_2$}}
\label{fig:Q2_motive1}
\end{subfigure}
\hspace*{5ex}
\begin{subfigure}[h]{0.4\linewidth}
\includegraphics[width=\linewidth]{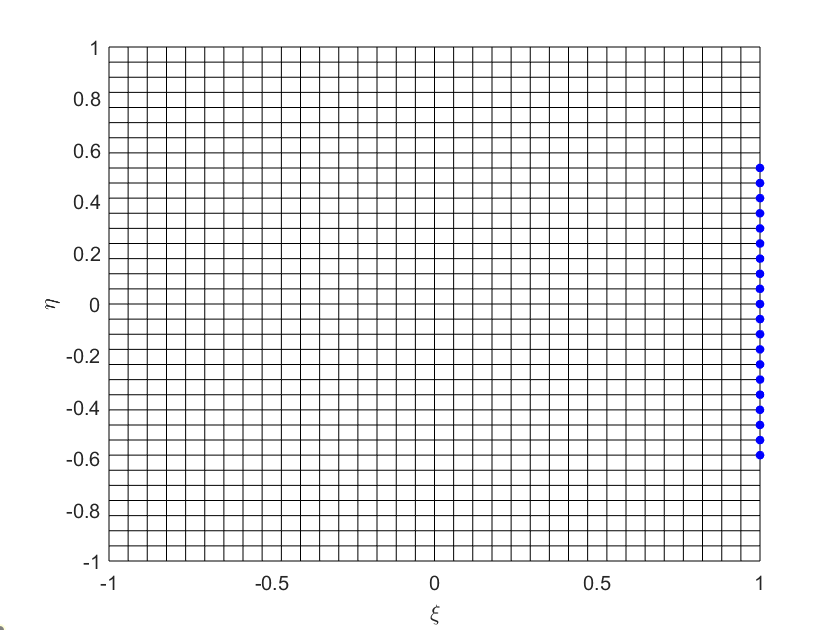}
\caption{{\small Dots are negative values ($\mathcal{Q}_2$)}}
\label{fig:Q2_motive2}
\end{subfigure}%
\caption{Plots of $v$ for $\mathcal{P}_2$ (Figure~\ref{fig:P2_motive1}) and $\mathcal{Q}_2$  (Figure~\ref{fig:Q2_motive1}). Figure~\ref{fig:P2_motive2} and \ref{fig:Q2_motive2} display locations where $v$ is negative (indicated with dots).}
\label{fig:Xk_motive}
\end{figure} 

One might hope that simply increasing the polynomial degree would satisfy the constraints, however, this does not appear to remedy the situation.  The minimum of $v$ arising from the standard spaces $\mathcal{P}_k$ and $\mathcal{Q}_k$ for various $k>1$ is displayed in Fig~\ref{fig:high_order_Pk_Qk}.  We again restrict our study to a single cell, with the same parameters $\alpha=1$, $\beta=1/10$, $\gamma=0$, $\Delta x_i=\Delta y_j = 0.1$.  From Fig~\ref{fig:high_order_Pk_Qk} it is clear that increasing the polynomial degree will not necessarily generate a non-negative test function $v$, therefore, this will not ensure a positive cell average.
 \begin{figure}[htb!]
\centering
\includegraphics[trim = 30mm 80mm 40mm 85mm, clip, scale=0.4]{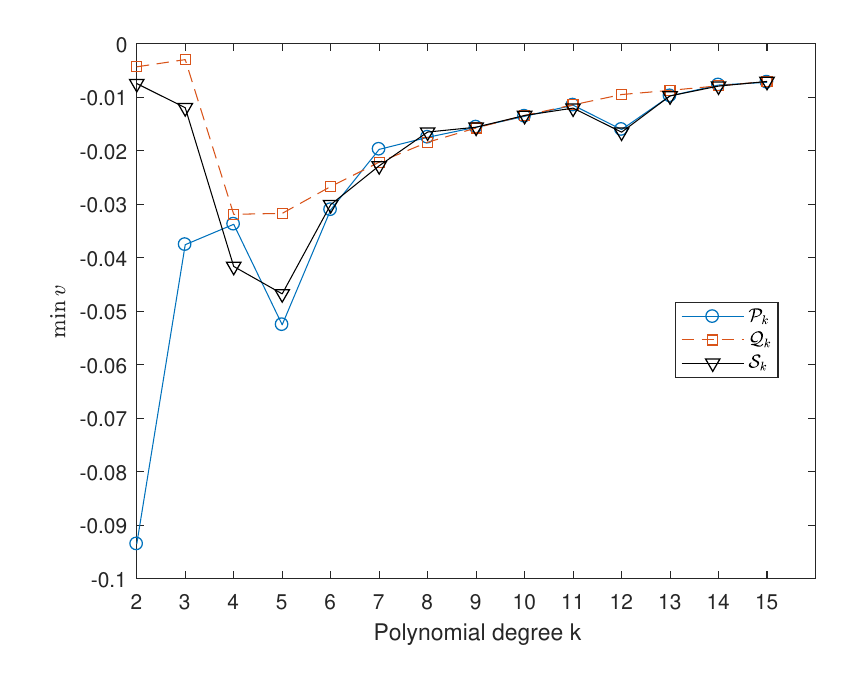}
\caption{Minimums for the basis function $v$ resulting from the unaugmented spaces $\mathcal{P}_k$, $\mathcal{Q}_k$, and $\mathcal{S}_k$ (Serendipity space) are plotted against their respective polynomial degree $k$.  The minimums are strictly negative}
\label{fig:high_order_Pk_Qk}
\end{figure} 
\subsection{Variable coefficients}
	Here we highlight a convenient feature of Problem~\eqref{eq:optimization_algorithm} - it can be used with non-negative variable coefficients.  In the case of variable coefficients, the augmented functions for $k=1$ suggested in~\cite{ling2018conservative} do not provide a non-negative test function $v$.  In particular, they augment $\mathcal{P}_1$ with the function 
\begin{align}
\psi(x,y) = xy + (1/4)(x^2+y^2)
\label{eq:basis_P1}
\end{align} 
However, this choice is only suitable for constant coefficients. 
	
	 We provide a counterexample.  Let $\alpha(x,y)= e^{-x-y}$, $\beta(x,y) = ( 1 + e^2 ) - e^{-x-y}$, and $\gamma(x,y) \equiv 0$.  Set the domain to be the reference cell $[-1,1]^2$, $\Delta x =\Delta y =0.1$, and $x=\xi$, $y=\eta$.  Note that $\alpha,\beta>0$ on the reference cell, and the vector field $[\alpha(x,y),\beta(x,y)]^T$ is divergence-free. 
	 
	 Fig.~\ref{fig:counter_ex0} plots the test function $v$ resulting from using the augmented spaces $\widetilde{\mathcal{P}}_1$ or $\widetilde{\mathcal{Q}}_1$.  
\begin{figure}[htb!]
\centering
\begin{subfigure}[h]{0.4\linewidth}
\includegraphics[width=\linewidth]{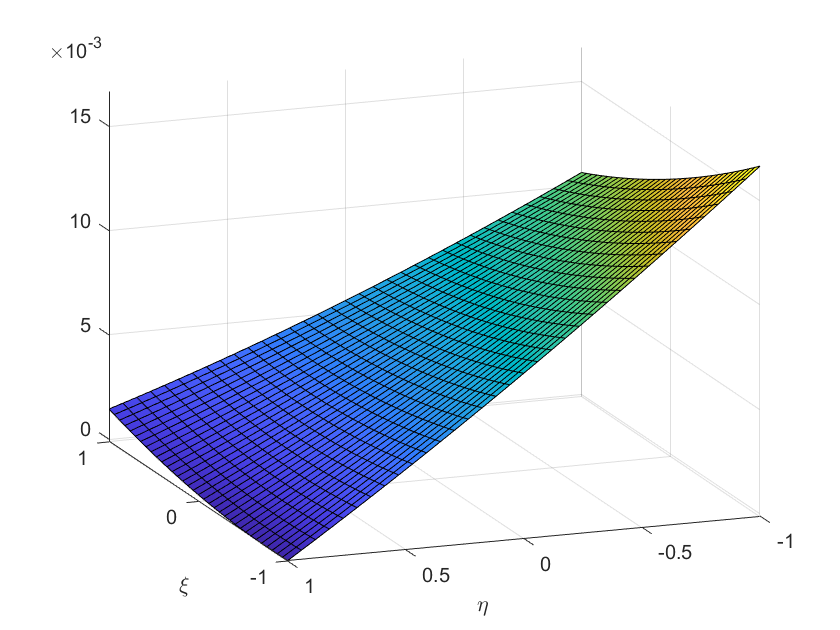}
\caption{Region plot of $v$ ($\widetilde{\mathcal{P}}_1$)}
\end{subfigure}
\hspace*{5ex}
\begin{subfigure}[h]{0.4\linewidth}
\includegraphics[width=\linewidth]{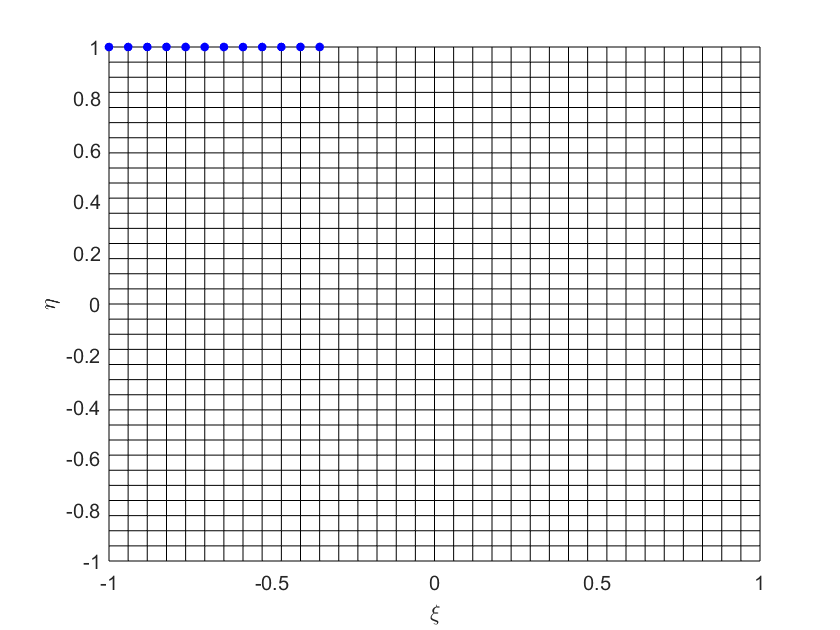}
\caption{Points where $v$ is negative ($\widetilde{\mathcal{P}}_1$)}
\end{subfigure}%
\\
\begin{subfigure}[h]{0.4\linewidth}
\includegraphics[width=\linewidth]{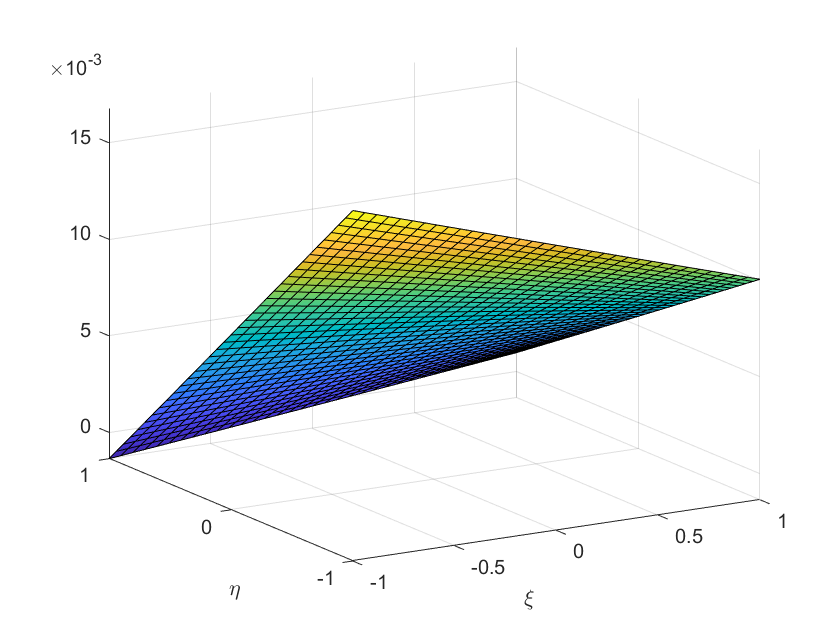}
\caption{Region plot of $v$ ($\widetilde{\mathcal{Q}}_1$)}
\end{subfigure}
\hspace*{5ex}
\begin{subfigure}[h]{0.4\linewidth}
\includegraphics[width=\linewidth]{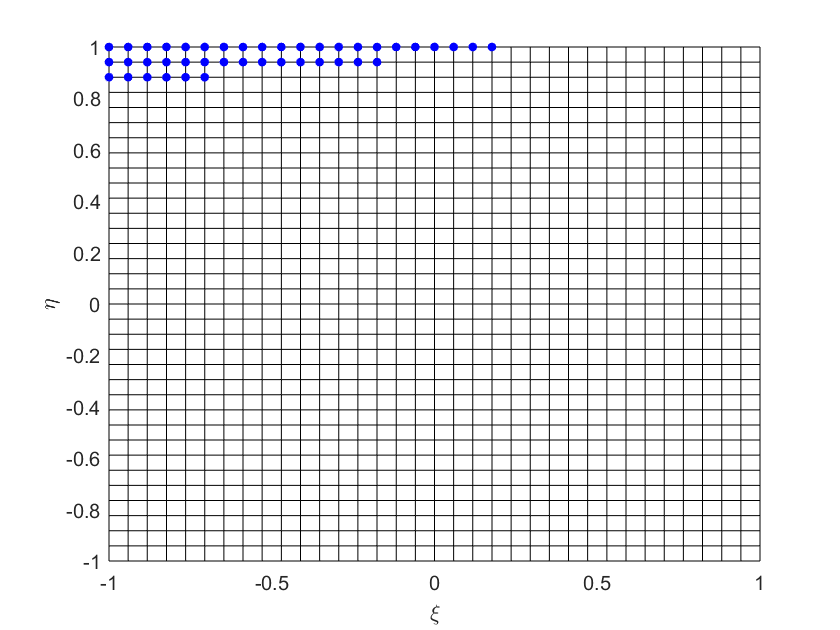}
\caption{Points where $v$ is negative ($\widetilde{\mathcal{Q}}_1$)}
\end{subfigure}%
\caption{Basis function $v$ resulting from the augmented spaces ($\psi$ from~\eqref{eq:basis_P1}, see~\cite{ling2018conservative}) $\widetilde{\mathcal{P}}_1$ or $\widetilde{\mathcal{Q}}_1$ are not non-negative.  Filled circles indicate locations where $v$ is negative}
\label{fig:counter_ex0}
\end{figure} 
		 We see that they do not remain non-negative.  Thus, the augmented basis function~\eqref{eq:basis_P1} does not ensure a positive cell average.  On the other hand, from Table~\ref{table_example_section_ex_variable}, it is evident that Problem~\eqref{eq:optimization_algorithm} is able to ensure that $v$ is non-negative by automatically finding a suitable augmented basis function.
		 { \small
\begin{table}[htb!]
  \centering 
  \begin{tabular}{cl}
    \toprule
 Augmented space              & $\min_{(\xi,\eta)\in [-1,1]^2} v(\xi,\eta)$  \\
    \midrule   
  $\widetilde{\mathcal{P}}_1$~($\psi$ from~\eqref{eq:basis_P1}, see~\cite{ling2018conservative})                 & $-1.268\mathrm{E-04}$    \\
  $\widetilde{\mathcal{Q}}_1$~($\psi$ from~\eqref{eq:basis_P1}, see~\cite{ling2018conservative})                 & $-1.363\mathrm{E-03}$    \\  
  $\widetilde{\mathcal{P}}_1$~(Problem~\eqref{eq:optimization_algorithm}) & $ 1.843\mathrm{E-04}$    \\ 
  $\widetilde{\mathcal{Q}}_1$~(Problem~\eqref{eq:optimization_algorithm}) & $ 1.155\mathrm{E-04}$    \\   
    \bottomrule
  \end{tabular}
  \caption{Minimums of the test function $v$ (acquired from different augmented spaces).  Problem~\eqref{eq:optimization_algorithm} automatically ensures $v$ is non-negative.
The augmented spaces from~\cite{ling2018conservative} result in a negative test function $v$}
  \label{table_example_section_ex_variable}
\end{table}	 
}

\subsection{Example augmented spaces} \label{sec:example_aug}
Although the primary feature of Problem~\eqref{eq:optimization_algorithm} is to implicitly generate suitable test and trial spaces on the fly, it is possible to use it to pre-process them (as apposed to activating Problem~\eqref{eq:optimization_algorithm} on troubled cells). In particular, in this section we explore explicit formulas for augmented basis functions which are valid for multiple cells (e.g, obeying CFL-type conditions).

We remark that the explicit augmented functions presented here are not unique or optimal, but are selected for their succinct representation. Problem~\eqref{eq:optimization_algorithm} generates $\dim \mathcal{X}_r$ real coefficients (housed in $\vec{d}$) that allow us to form the augmented function candidate $\psi = \sum_{i=1}^{ \dim{\mathcal{X}_r} } \vec{d}_i \Phi_i$. As $r>k$, not only will there be a proliferation of these coefficients, but their decimal expansions are unwieldy in general.

Let $\alpha,\beta,\gamma$ be constants. Proving that the augmented basis function $\psi$ gives rise to a viable augmented space is symbolically intensive, because $v$ is obtained a-posteriori via solving a linear system on the reference cell. State-of-the-art symbolic software such as Mathematica, Maxima, and Maple are unable to verify non-negativity analytically in all cases of admissible parameters. In this section we present numerical evidence that demonstrates $v(x,y) \ge 0$ and satisfies~\eqref{eq:LHS_avg} (see  Section \ref{section:Q2_proof}, Appendix A, for analytic evidence).

	Computations are focused on the reference cell, denoted by $\hat{S}=[-1,1]^2$. To reduce the parameter space, we scale the equation in~\eqref{eq:LHS_avg}. Mapping~\eqref{eq:LHS_avg} to the reference cell, and scaling the resulting equation by $(4/(\Delta x \Delta y))(\Delta x/(2\alpha))$ gives
\begin{subequations}
\begin{align}
\frac{\Delta x  }{2 \alpha} 
 \iint_{\hat{S}} u \, d\xi d\eta
&=
-  \iint_{\hat{S}} u \bigg(   v_{\xi} + \frac{\Delta x \beta}{\Delta y \alpha}   v_{\eta} \bigg)   \, d\xi d\eta  
+
\frac{\Delta x \beta}{\Delta y \alpha} \int_{\eta=1} \  u(\xi,1) v(\xi,1) \,d\xi
\\
&+
 \int_{\xi=1} \  u(1,\eta) v(1,\eta) \,d\eta
+
 \frac{\Delta x}{2\alpha} \iint_{\hat{S}} \gamma u v   \, d\xi d\eta  
\end{align}
\end{subequations}
Define
\begin{align} 
B = \frac{\Delta x \beta}{\Delta y \alpha},~~G =  \frac{\Delta x \gamma}{2\alpha}.
\label{eq_CFL_definition}
\end{align}
The scalar $B$ enables us to relate our computations to CFL-type conditions~\cite{courant1967partial,leveque1992numerical}.  Without a loss of generality, we assume that $\alpha= 1$, since it can be interpreted as scaling the model problem~\eqref{eq:model} by $1/\alpha$ and then applying the DG scheme to the scaled equation. 

Therefore, with the above modifications and assumption on $\alpha$, the parameter space depends on the triplet $(\Delta x,B,G)$ (as opposed to 5 parameters $(\Delta x,\Delta y, \beta,\alpha,\gamma)$).  In the following we assume $ (\xi,\eta) \in \hat{S}=[-1,1]^2$.

\subsection*{Augmenting $\mathcal{Q}_k$ in (1+1) dimensions} 
\subsubsection*{Case for $k=2$, $B< 0.5$}
Using Problem~\eqref{eq:optimization_algorithm}, the following basis function is obtained
\[
\psi_{1,\mathcal{Q}_2}(\xi,\eta) = 
\bigg(
 \xi \eta^2(1-\xi)^2(\xi + 1)(1-\eta^2) /8
\bigg)^2.
\]
The associated augmented space $\widetilde{\mathcal{Q}}_2 = \mathcal{Q}_2 \cup \text{span}(\psi_{1,\mathcal{Q}_2})$ is numerically found to ensure $v\ge0$ for any $B\le 0.5$, $G \ge 0$.
 
\subsubsection*{Case for $k=2$, $B> 2$}
The augmented basis function can be designed so that positivity is ensured for arbitrarily large CFL conditions. Using Problem~\eqref{eq:optimization_algorithm}, the following basis function is obtained
\[
\psi_{2,\mathcal{Q}_2}(\xi,\eta) = 
\bigg(
  \eta \xi^2 (1-\eta)^2 (\eta + 1)   (1-\xi^2)/8
\bigg)^2.
\]
The associated augmented space $\widetilde{\mathcal{Q}}_2 = \mathcal{Q}_2 \cup \text{span}(\psi_{2,\mathcal{Q}_2})$ is numerically found to ensure $v\ge0$ for any $B\ge 2$, $G \ge 0$.
 
Table~\eqref{table_example_section_ex_Q2} exhibits explicit augmented basis functions for $\mathcal{Q}_k$, with various $k$. These augmented basis functions are found to be computationally valid for all cells satisfying a specific CFL type condition. The augmented functions in this table are not unique, but are selected for their somewhat straightforward expressions. 
\begin{table}[htb!] 
  \centering  
  \resizebox{\linewidth}{!}{%
  \begin{tabular}{clc}
    \toprule
 Condition to enforce $\overline{u}_{ij}\ge0$ & Augmented basis function $\psi$                        & Original space \\
    \midrule  
  $(\beta/\alpha)\Delta x/\Delta y < \frac{1}{2}$ & $\psi_{1,\mathcal{Q}_2} =  (xy^2(1-x)^2(x + 1)(1-y^2)/8)^2 $ & $\mathcal{Q}_2$   \\
  $(\beta/\alpha)\Delta x/\Delta y > 2$           & $\psi_{2,\mathcal{Q}_2} = (y x^2 (1-y)^2 (y + 1)   (1-x^2)/8)^2$ & $\mathcal{Q}_2$\\  
  $(\beta/\alpha)\Delta x/\Delta y < \frac{1}{4}$ & $\psi_{1,\mathcal{Q}_3}  $ (see \eqref{eq:Appendix_Q3_psi_1})
  & $\mathcal{Q}_3$   \\ 
  $(\beta/\alpha)\Delta x/\Delta y > 4$           & $\psi_{2,\mathcal{Q}_3}  $ (see \eqref{eq:Appendix_Q3_psi_2}) & $\mathcal{Q}_3$\\   
  
  $(\beta/\alpha)\Delta x/\Delta y < \frac{1}{8}$ & $\psi_{1,\mathcal{Q}_4}  $ (see \eqref{eq:Appendix_Q4_psi_1}) & $\mathcal{Q}_4$   \\ 
  $(\beta/\alpha)\Delta x/\Delta y > 8$           & $\psi_{2,\mathcal{Q}_4}  $ (see \eqref{eq:Appendix_Q4_psi_2}) & $\mathcal{Q}_4$\\     
    \bottomrule
  \end{tabular}}
  \caption{Example augmented basis functions for $\mathcal{Q}_k$ generated by Problem~\ref{eq:optimization_algorithm}. Appendix B has explicit formulas for select higher-order augmented functions ($\mathcal{Q}_k$)} 
  \label{table_example_section_ex_Q2}  
\end{table}
 
The robustness of these augmented basis functions with respect to the CFL condition is examined.  Consider a single cell $S_{ij}$.  Random values are generated for $\Delta x_i,\Delta y_j,\alpha,\beta,\gamma$ to use in Problem~\eqref{eq:optimization_algorithm}.  Define $B:=\mathrm{CFL}= (\beta/\alpha)(\Delta x_i/\Delta y_j)$.  For several polynomial degrees, Fig.~\ref{fig:cfl_ex} plots the minimum of the test function $v$ obtained from Table~\eqref{table_example_section_ex_Q2} and Problem~\eqref{eq:optimization_algorithm} against a wide range of CFL conditions.
\begin{figure}[htb!]
\centering
\includegraphics[trim = 35mm 80mm 40mm 85mm, clip, scale = 0.5]{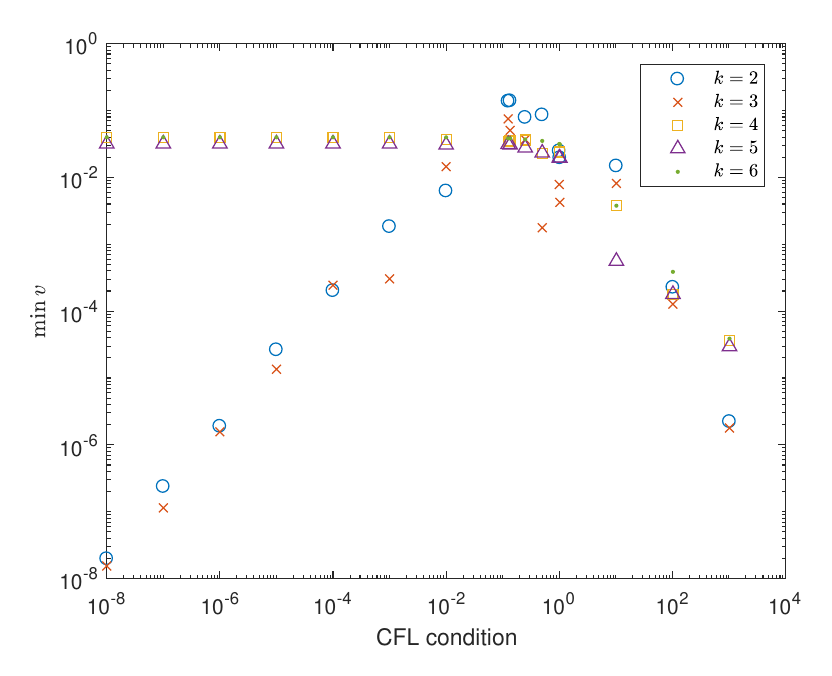}
\caption{Minimum of $v$ vs CFL number ($\widetilde{\mathcal{Q}}_k$). Log-log plot of the minimum of $v$ generated from Problem~\eqref{eq:optimization_algorithm} vs CFL number.  The test functions $v$ are non-negative for a wide range of CFL numbers}
\label{fig:cfl_ex}
\end{figure}
Our approach augments a single basis function to $\mathcal{Q}_2$.  As such, the space ${\widetilde{ \mathcal{ Q} }}_2$ will have ten degrees of freedom, which is slightly more expensive than the unaugmented space.  To arrive at a third order accurate scheme that has cost comparable to $\mathcal{Q}_2$ (and has a positive cell average), one could consider augmenting the eight-node serendipity element instead. This way one would have nine degrees of freedom (same as $\mathcal{Q}_2$), and the accuracy will be of third order~\cite{arnold2002approximation}.

\subsection*{Augmenting $\mathcal{S}_k$ in (1+1) dimensions} 
 To reduce the computational cost, but retain $k+1$ order accuracy, one could consider augmenting the serendipity spaces $\mathcal{S}_k$. Table~\eqref{table_example_section_ex_S2} displays several explicit augmented basis functions which are valid under certain CFL conditions. The spaces $\widetilde{\mathcal{S}}_k = \mathcal{S}_k \cup \text{span}\{\psi_{\cdot ,\mathcal{S}_k}\}$ enforce the existence of a non-negative test function under these CFL conditions.
\begin{table}[htb!]
  \centering 
  \resizebox{\linewidth}{!}{%
  \begin{tabular}{clc}
    \toprule
 Condition to enforce $\overline{u}_{ij}\ge0$ & Augmented basis function $\psi$                        & Original space \\
    \midrule  
  $(\beta/\alpha)\Delta x/\Delta y < \frac{1}{2}$ & $\psi_{1,\mathcal{S}_2} = 
  (
\xi  \eta^2  (\xi - 1)^2  (\eta - 1)    (\xi + 1)  (\eta + 1)/8
  )^2$ 
  & $\mathcal{S}_2$   \\
  $(\beta/\alpha)\Delta x/\Delta y > 2$           & 
  $\psi_{2,\mathcal{S}_2} 
  = (
\xi^2  \eta  (\xi - 1)  (\eta - 1)^2 
  (\xi + 1)  (\eta + 1)/2
  )^2$
   & $\mathcal{S}_2$\\     
  $(\beta/\alpha)\Delta x/\Delta y <\frac{1}{8}$ & $\psi_{1,\mathcal{S}_3} $ (see \eqref{eq:Appendix_S3_psi_1}) & $\mathcal{S}_3$   \\ 
  $(\beta/\alpha)\Delta x/\Delta y > 8$           & $\psi_{2,\mathcal{S}_3}  $ (see \eqref{eq:Appendix_S3_psi_2}) & $\mathcal{S}_3$\\   
  
  $(\beta/\alpha)\Delta x/\Delta y < \frac{1}{16}$ & $\psi_{1,\mathcal{S}_4} $ (see \eqref{eq:Appendix_S4_psi_1}) & $\mathcal{S}_4$   \\ 
  $(\beta/\alpha)\Delta x/\Delta y > 16$           & $\psi_{2,\mathcal{S}_4}  $ (see \eqref{eq:Appendix_S4_psi_2}) & $\mathcal{S}_4$\\    
    \bottomrule
  \end{tabular}
  }
  \caption{Example augmented basis functions for the serendipity spaces $\mathcal{S}_k$ generated by Problem~\ref{eq:optimization_algorithm}. Appendix C has explicit formulas for select higher-order augmented functions ($\mathcal{S}_k$)
  } 
  \label{table_example_section_ex_S2}
\end{table}
Fig.~\eqref{fig:cfl_ex3} demonstrates the effectiveness of the augmented basis functions. When augmented to the standard serendipity spaces, the minimum of $v$ remains non-negative for CFL conditions arbitrarily small or large.
\begin{figure}[htb!]
\centering
\includegraphics[trim = 35mm 80mm 40mm 85mm, clip, scale = 0.5]{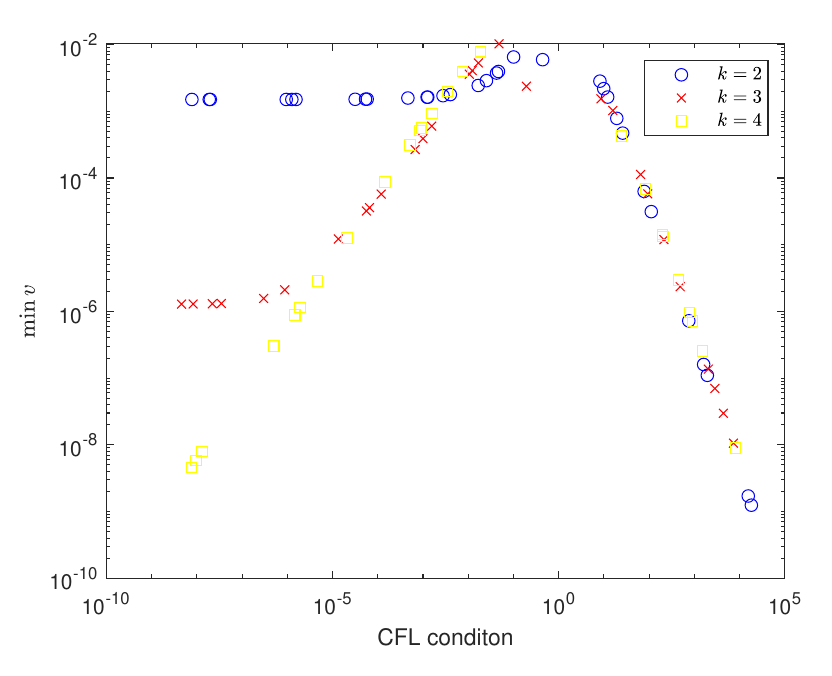}
\caption{Minimum of $v$ vs CFL number ($\widetilde{\mathcal{S}}_k$). Log-log plot of the minimum of $v$ generated from Problem~\eqref{eq:optimization_algorithm} vs CFL number.  The test functions $v$ are non-negative for a wide range of CFL numbers} 
\label{fig:cfl_ex3}
\end{figure}
We emphasize that the augmented basis functions in Tables \ref{table_example_section_ex_Q2} and \ref{table_example_section_ex_S2} are determined computationally, and theoretical evidence for validity for all parameters $\alpha,\beta,\gamma$ is lacking. However, Appendix A provides theoretical evidence for special cases, and numerical experiments give good results.

\subsection*{Augmenting $\mathcal{P}_k$ in (1+1) dimensions} 
Problem~\ref{eq:optimization_algorithm} can successfully find augmented basis functions for $\mathcal{P}_k$ as well. Numerically we investigate the robustness of Problem~\ref{eq:optimization_algorithm} by generating random values for $\Delta x_i,\Delta y_j,\alpha,\beta,\gamma$. For each polynomial degree, we inspect the minimum of $v$ resulting from the associated augmented spaces. 
\begin{figure}[htb!] 
\centering
\includegraphics[trim = 35mm 80mm 40mm 85mm, clip, scale = 0.5]{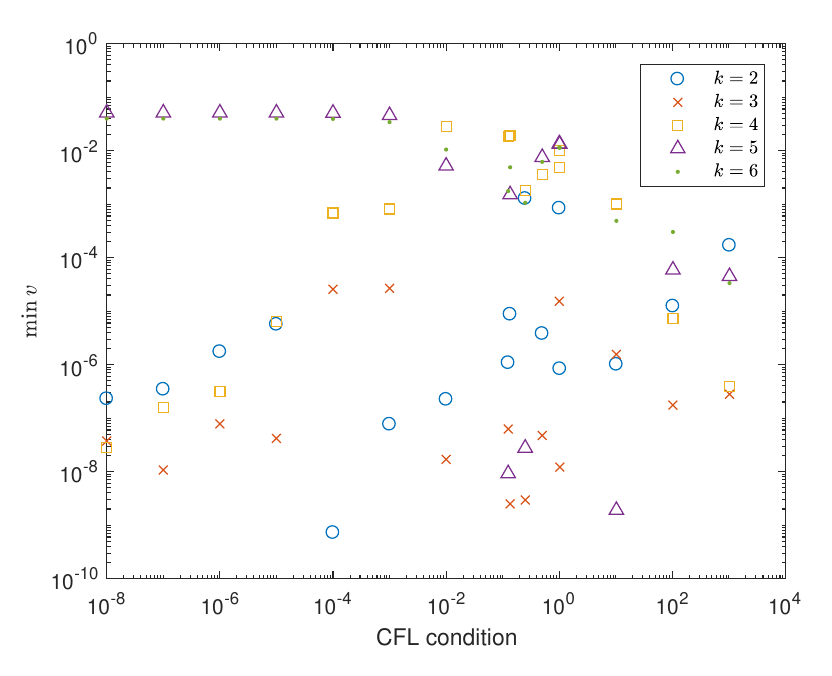}
\caption{Minimum of $v$ vs CFL number ($\widetilde{\mathcal{P}}_k$).  Log-log plot of the minimum of $v$ generated from Problem~\eqref{eq:optimization_algorithm} vs CFL number.  The test functions $v$ are non-negative for a wide range of CFL numbers}
\label{fig:cfl_ex2}
\end{figure}
 Fig.~\ref{fig:cfl_ex2} plots the minimum of the test function $v$ against CFL numbers. It is evident that Problem~\ref{eq:optimization_algorithm} is creating non-negative test functions. 
 
\section{Numerical experiments} \label{section:numerics}
This section conducts several numerical experiments which verify and validate the proposed high-order accurate positivity-preserving DG scheme.  Computational results are conducted in $(1+1)$ dimensions and $(1+2)$ dimensions.  The numerical experiments provided here assume that it is understood from context that $x$ is replaced with $t$ and $y=x$ in~\eqref{eq:model} and~\eqref{eq:dg_scheme}.  Similar test cases can be found in~\cite{ling2018conservative}.  However, we remark that in~\cite{ling2018conservative} the backward Euler method is used to resolve the time derivative of~\eqref{eq:model}, which reduces the overall accuracy to first order.  Our method provided in this paper is $k+1$ order accurate in space and time.

\subsection{Positivity violation for the \texorpdfstring{$\mathcal{S}_2$}{S2} DG scheme in (1+1) dimensions} \label{section:numerics_S2_example}
In this section we examine an example which demonstrates that the $\mathcal{  S}_2$ DG scheme will not preserve a positive cell average in general; even when given a non-negative source and inflow boundary conditions.  The computational domain is a single cell for this experiment.  In particular, the following space-time domain is selected: $(x,t)\in  [0,\Delta x]\times [0,\Delta t].$     We take $\alpha=2,\beta=1$, $\gamma=0$, and 
\begin{align*}
 u(x,t) &=  \begin{cases} 
        20t + x, & \text{for }  \frac{487}{496} \le \xi \le 1,~ -\frac{17}{18}\le \eta \le \frac{60\xi-67}{24\xi-15}   \\
        0, & \text{otherwise} 
        \end{cases} 
\\
 f(x,t) &=  \begin{cases} 
        41, & \text{for } \frac{487}{496} \le \xi \le 1,~ -\frac{17}{18}\le \eta \le \frac{60\xi-67}{24\xi-15} \\
        0, & \text{otherwise} 
        \end{cases} 
.
\end{align*}
Note that the forcing function $f$ and $u$ are defined by the mapping from the reference square $(\xi,\eta)\in [-1,1]^2.$  For this example $(\beta/\alpha)\Delta t / \Delta x = 1/2$, and as such, we select the augmented basis function given in the second row of Table~\ref{table_example_section_ex_S2}.
\begin{table}[htb!]
  \centering 
  \begin{tabular}{lccc}
    \toprule
$\Delta t = \Delta x$  & Exact  & $\mathcal{S}_2$ & $\widetilde{\mathcal{ S}}_2$\\
    \midrule  
$1/8$  &$9.4\mathrm{E-04}$ &$-2.5\mathrm{E-05}$ &$1.9\mathrm{E-05}$\\
$1/16$ &$4.7\mathrm{E-04}$ &$-1.2\mathrm{E-05}$ &$9.8\mathrm{E-06}$\\
$1/32$ &$2.3\mathrm{E-04}$ &$-6.4\mathrm{E-06}$ &$4.9\mathrm{E-06}$\\
$1/64$ &$1.1\mathrm{E-04}$ &$-3.2\mathrm{E-06}$ &$2.4\mathrm{E-06}$\\
    \bottomrule
  \end{tabular}
  \caption{Cell averages $\overline{u}_{1,1}$ generated by the $\mathcal{S}_2$ DG and $\widetilde{\mathcal{ S}}_2$ DG schemes}
  \label{table_example_section_pos_vio_S2}
\end{table}
From Table~\ref{table_example_section_pos_vio_S2}, we see that the unaugmented $\mathcal{S}_2$ DG scheme violates positivity of the cell average, even when the mesh spacing is reduced.  On the other hand, the augmented $\widetilde{\mathcal{ S}}_2$ DG scheme preserves a positive cell average.

\subsection{Positivity violation for the \texorpdfstring{$\mathcal{Q}_2$}{Q2} DG scheme in (1+1) dimensions}
The numerical experiment in this section is analogous to the one described in Section~\ref{section:numerics_S2_example}.  We demonstrate that the $\mathcal{Q}_2$ DG scheme will not preserve a positive cell average in general; even when given a non-negative source and inflow boundary conditions.  
The computational domain is again taken to be a single cell for this experiment.  In particular, the following space-time domain is selected: $(x,t)\in  [0,\Delta x]\times [0,\Delta t].$  We take $\alpha=\beta=1$, $\gamma=0$, and 
\begin{align*}
 u(x,t) &=  \begin{cases} 
        200x + t, & \text{for }  \frac{487}{496} \le \xi \le 1,~ -\frac{17}{18}\le \eta \le \frac{60\xi-67}{24\xi-15} \\
        0, & \text{otherwise} 
        \end{cases} 
\\
 f(x,t) &=  \begin{cases} 
        201, & \text{for }  \frac{487}{496} \le \xi \le 1,~ -\frac{17}{18}\le \eta \le \frac{60\xi-67}{24\xi-15} \\
        0, & \text{otherwise} 
        \end{cases} 
.
\end{align*}
Note the forcing function $f$ and $u$ are defined by the mapping from the reference square $(\xi,\eta)\in [-1,1]^2.$  We select $\Delta t= \Delta x$.  It is straightforward to verify that $(\beta/\alpha)\Delta t / \Delta x = 1$, so $\widetilde{\mathcal{ Q}}_2$ results from augmenting $\mathcal{Q}_2$ with the associated $\psi$ from Table~\ref{table_example_section_ex_Q2}.  Table~\ref{table_example_section_pos_vio_Q2} records the cell average $\overline{u}_{1,1}$ for both the $\mathcal{Q}_2$ DG and $\widetilde{\mathcal{ Q}}_2$ DG schemes.
\begin{table}[htb!]
  \centering 
  \begin{tabular}{lccc}
    \toprule
$\Delta t = \Delta x$  & Exact  & $\mathcal{Q}_2$ & $\widetilde{\mathcal{ Q}}_2$\\
    \midrule   
$1/8$   &$9.5\mathrm{E-03}$    &$-3.5\mathrm{E-04}$ &$2.6\mathrm{E-05}$\\
$1/16$  &$4.7\mathrm{E-03}$    &$-1.7\mathrm{E-04}$ &$1.3\mathrm{E-05}$\\
$1/32$  &$2.3\mathrm{E-03}$    &$-8.8\mathrm{E-05}$ &$6.5\mathrm{E-06}$\\
$1/64$  &$1.1\mathrm{E-03}$    &$-4.4\mathrm{E-05}$ &$3.2\mathrm{E-06}$\\
$1/128$ &$5.9\mathrm{E-04}$    &$-2.2\mathrm{E-05}$ &$1.6\mathrm{E-06}$\\
    \bottomrule
  \end{tabular}
  \caption{Cell averages $\overline{u}_{1,1}$ generated by the $\mathcal{Q}_2$ DG and $\widetilde{\mathcal{ Q}}_2$ DG schemes}
  \label{table_example_section_pos_vio_Q2}
\end{table}
It is evident that the $\mathcal{Q}_2$ DG scheme generates a negative cell average.  However, the $\widetilde{\mathcal{ Q}}_2$ DG scheme retains a positive cell average.

\subsection{Positivity violation for the \texorpdfstring{$\mathcal{P}_2$}{P2} DG scheme in (1+1) dimensions}  
Higher-order $\mathcal{P}_k$ spaces also do not guarantee an unmodulated positive cell average.  To see this, we present a counter example for $k=2$.  Take $\alpha = \beta=1$, $\gamma=0$ in~\eqref{eq:model}, and put the exact solution as 
\[
u(x,t) = \bigg( \frac{31}{2}  x t \bigg)^{13},
\quad
(x,t) \in [0,0.5]^2 
.
\]  
The computational domain $[0,0.5]^2$ is partitioned uniformly, $N_x$ rectangles in the $x$ direction, and $N_t$ rectangles in the $t$ direction. 
\begin{table}[htb!]
  \centering 
  \begin{tabular}{lccc}
    \toprule
$N_x = N_t$  & Exact  & $\mathcal{P}_2$ & $\widetilde{\mathcal{ P}}_2$\\
    \midrule   
$2$   &$3.37\mathrm{E-03}$    &$-1.09\mathrm{E-04}$ &$1.47\mathrm{E-03}$\\
$4$   &$5.03\mathrm{E-11}$    &$-1.62\mathrm{E-12}$ &$2.19\mathrm{E-11}$\\
$8$   &$7.49\mathrm{E-19}$    &$-2.42\mathrm{E-20}$ &$3.27\mathrm{E-19}$\\
    \bottomrule
  \end{tabular}
  \caption{Cell averages $\overline{u}_{1,1}$ generated by the $\mathcal{P}_2$ DG and $\widetilde{\mathcal{ P}}_2$ DG schemes}
  \label{table_example_section_pos_vio_P2}
\end{table}  
Using Problem~\ref{eq:optimization_algorithm} $(r=k+2)$, we obtain $\psi(x,y) = (1/2)(xt)^2$, which is used to augment the space $\mathcal{P}_2$.  We note that this specific $\psi$ is only valid for the CFL condition
 \[
 \frac{(\beta/\alpha)  \Delta t}{\Delta x  } = 1.
 \]
We track the cell average in the bottom left cell, $S_{11}=[0,\Delta x]\times[0,\Delta t]$.  Table~\ref{table_example_section_pos_vio_P2} displays the cell average in the cell $S_{11}$ for the $\mathcal{P}_2$ DG and $\widetilde{\mathcal{ P}}_2$ DG schemes.  The exact cell average is also provided.  It is clear that the $\mathcal{P}_2$ DG is unable to preserve a positive cell average.  When the $\widetilde{\mathcal{ P}}_2$ DG scheme is used, the cell average remains non-negative.
  
\subsection{Accuracy test for the DG schemes in (1+1) dimensions} \label{section:numerics_accuracy_example}
A manufactured solution is used to verify numerically that the correct convergence rates are obtained.  We take $\alpha=\beta=\gamma=1$, and $u(x,t) = \cos^4{(x-t)}+10^{-14}$, with $\mathcal{D}=[0,0.1]\times[0,2\pi]$.  The source $f(x,t) $ can be determined from~\eqref{eq:model} using the exact solution $u$.  We see that $u\ge0$ and $f\ge0$.  Therefore, the inflow boundary conditions $u(x,0)$ and $u(0,t)$ are non-negative.  The time domain $[0,0.1]$ is divided into $N_t$ subintervals, and the space domain $[0,2\pi]$ is divided into $N_x$ subintervals.

Several spaces $\mathcal{X}_k$ are verified with the DG scheme, including their augmented counterparts $\mathcal{\widetilde{X}}_k$.  For reproducibility purposes, we use the augmented basis functions described in Appendix B and Appendix C.  Errors and orders of accuracy are given in Table~\ref{table_example_section_manufactured}.  The schemes give space-time orders of accuracy of $k+1$ in the $L^2$ norm. 

\begin{table}[htb!]
  \centering 
  \begin{tabular}{lcccccc}
    \toprule
     & &   \multicolumn{3}{c}{$\mathcal{X}_k$}  & \multicolumn{2}{c}{$ \mathcal{\widetilde{X}}_k$}\\
      \cmidrule(lr){3-5}\cmidrule(lr){6-7}
                              &$N_x=N_t$  & $L^2$ error  & $L^2$ order  &  $\min u_h$ & $L^2$ error & $L^2$ order \\
    \midrule   
$\mathcal{X}_k=\mathcal{Q}_2$ &$10$  &$2.01\mathrm{E-03}$    &$-$    &$2.74\mathrm{E-04}$ &$2.02\mathrm{E-03}$ &$-$    \\
				              &$20$  &$2.97\mathrm{E-04}$    &$2.76$ &$-2.93\mathrm{E-04}$ &$2.99\mathrm{E-04}$ & 2.75 \\
							  &$40$  &$4.29\mathrm{E-05}$    &$2.79$ &$-1.76\mathrm{E-05}$ &$4.32\mathrm{E-05}$ & 2.79 \\
				              &$80$  &$5.31\mathrm{E-06}$    &$3.01$ &$-2.98\mathrm{E-06}$ &$5.35\mathrm{E-06}$ & 3.01 \\
				              &$160$ &$6.63\mathrm{E-07}$    &$3.00$ &$-1.87\mathrm{E-07}$ &$6.69\mathrm{E-07}$ & 2.99 \\
				              \\
$\mathcal{X}_k=\mathcal{S}_2$ &$10$  &$2.01\mathrm{E-03}$    &$-$    & $ 2.85\mathrm{E-04}$&$2.02\mathrm{E-03}$ &$-$   \\
				              &$20$  &$2.97\mathrm{E-04}$    &$2.76$ & $-2.93\mathrm{E-04}$&$2.97\mathrm{E-04}$ & 2.76 \\
							  &$40$  &$4.29\mathrm{E-05}$    &$2.79$ & $-1.76\mathrm{E-05}$&$4.28\mathrm{E-05}$ & 2.79 \\
				              &$80$  &$5.31\mathrm{E-06}$    &$3.01$ & $-2.98\mathrm{E-06}$&$5.29\mathrm{E-06}$ & 3.01 \\
				              &$160$ &$6.63\mathrm{E-07}$    &$3.00$ & $-1.87\mathrm{E-07}$&$6.61\mathrm{E-07}$ & 3.00 \\
				              \\
$\mathcal{X}_k=\mathcal{P}_2$ &$10$  &$2.02\mathrm{E-03}$    &$-$    & $ 4.61\mathrm{E-04}$&$2.00\mathrm{E-03}$ &$-$   \\
				              &$20$  &$2.98\mathrm{E-04}$    &$2.75$ & $-3.59\mathrm{E-04}$&$2.94\mathrm{E-04}$ & 2.76 \\
							  &$40$  &$4.29\mathrm{E-05}$    &$2.79$ & $-2.32\mathrm{E-05}$&$4.23\mathrm{E-05}$ & 2.79 \\
				              &$80$  &$5.31\mathrm{E-06}$    &$3.01$ & $-1.46\mathrm{E-06}$&$5.28\mathrm{E-06}$ & 3.01 \\
				              &$160$ &$6.64\mathrm{E-07}$    &$3.00$ & $-9.17\mathrm{E-08}$&$6.62\mathrm{E-07}$ & 2.99 \\
				              \\
$\mathcal{X}_k=\mathcal{Q}_3$ &$10$  &$1.68\mathrm{E-04}$    &$-$    & $ 2.63\mathrm{E-04}$&$1.67\mathrm{E-04}$ &$-$   \\
				              &$20$  &$1.18\mathrm{E-05}$    &$3.83$ & $-1.44\mathrm{E-05}$&$1.17\mathrm{E-05}$ & 3.83 \\
							  &$40$  &$8.32\mathrm{E-07}$    &$3.82$ & $-8.58\mathrm{E-07}$&$8.31\mathrm{E-07}$ & 3.82 \\
				              &$80$  &$5.18\mathrm{E-08}$    &$4.00$ & $-6.24\mathrm{E-08}$&$5.17\mathrm{E-08}$ & 4.00 \\
				              &$160$ &$3.26\mathrm{E-09}$    &$3.98$ & $-4.60\mathrm{E-09}$&$3.26\mathrm{E-09}$ & 3.98 \\
								\\
$\mathcal{X}_k=\mathcal{P}_4$ &$10$  &$1.12\mathrm{E-05}$    &$-$    & $ 8.84\mathrm{E-07}$&$1.11\mathrm{E-05}$ & $-$  \\
				              &$20$  &$3.94\mathrm{E-07}$    &$4.83$ & $-8.77\mathrm{E-08}$&$3.89\mathrm{E-07}$ & 4.83 \\
							  &$40$  &$1.27\mathrm{E-08}$    &$4.94$ & $-1.44\mathrm{E-09}$&$1.26\mathrm{E-08}$ & 4.94 \\
				              &$80$  &$4.01\mathrm{E-10}$    &$4.99$ & $-4.27\mathrm{E-11}$&$3.99\mathrm{E-10}$ & 4.98 \\
    \bottomrule
  \end{tabular}
  \caption{Manufactured solution example in 2D~(Section~\ref{section:numerics_accuracy_example}).  Errors and rates for the DG schemes with different spaces $\mathcal{X}_k$ and $ \mathcal{\widetilde{X}}_k$  }
  \label{table_example_section_manufactured}
\end{table}

\subsection{Step function propagation for the linear problem in (1+1) dimensions}
This numerical experiment propagates a step function.  The time interval is $[0,0.1]$ (divided into $N_t=40$ subintervals), and the space domain is $[0,2\pi]$ (divided into $N_x=40$ subintervals).  We take $\alpha=\beta=1$, $f\equiv 0$, and $\gamma=0$.  The initial condition is
\[
u(x,0)
=
 \begin{cases} 
        1, & \text{for } x\in[3,4] \\
        0, & \text{otherwise} .
        \end{cases} 
\]
The boundary condition is $u(0,t)=0$.  Fig.~\ref{fig:shock_ex1} plots a comparison between the exact solution and the numerical approximations for the various higher-order DG schemes (with and without the positivity-preserving limiter).  As expected, the DG schemes without the limiter produce approximations that have negative values.  Non-negative approximations result when the limiter is applied.  Fig.~\ref{fig:shock_ex1_zoom} has a zoom-in near the discontinuities to give the reader a better idea of the oscillations that occur due to the higher-order DG approximation.

\begin{figure}[htb!]
\centering
\begin{subfigure}[h]{0.49\linewidth}
\includegraphics[trim = 40mm 80mm 40mm 85mm, clip, width=\linewidth]{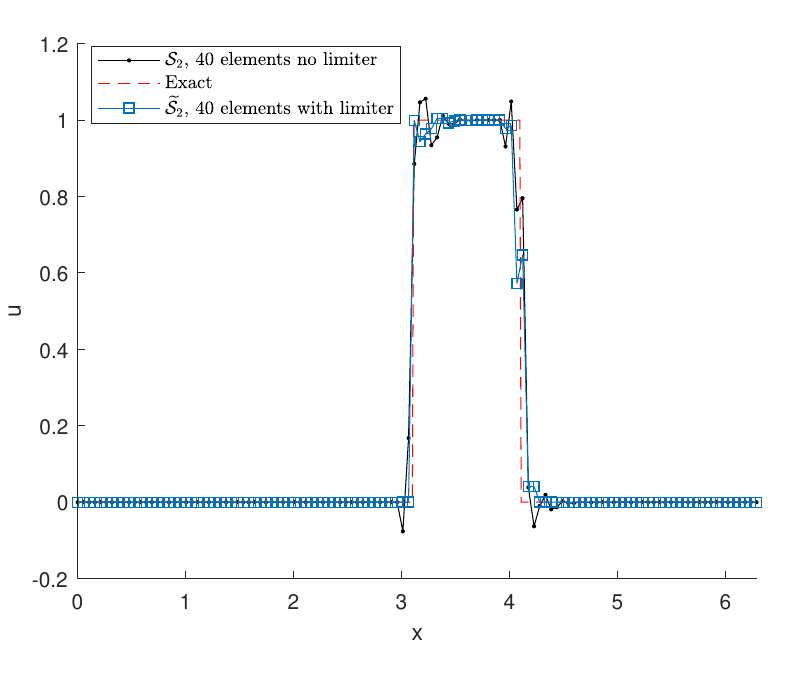}
\caption{$\mathcal{S}_2$ and $\mathcal{\widetilde{S}}_2$ DG schemes}
\end{subfigure}
\begin{subfigure}[h]{0.49\linewidth}
\includegraphics[trim = 40mm 80mm 40mm 85mm, clip, width=\linewidth]{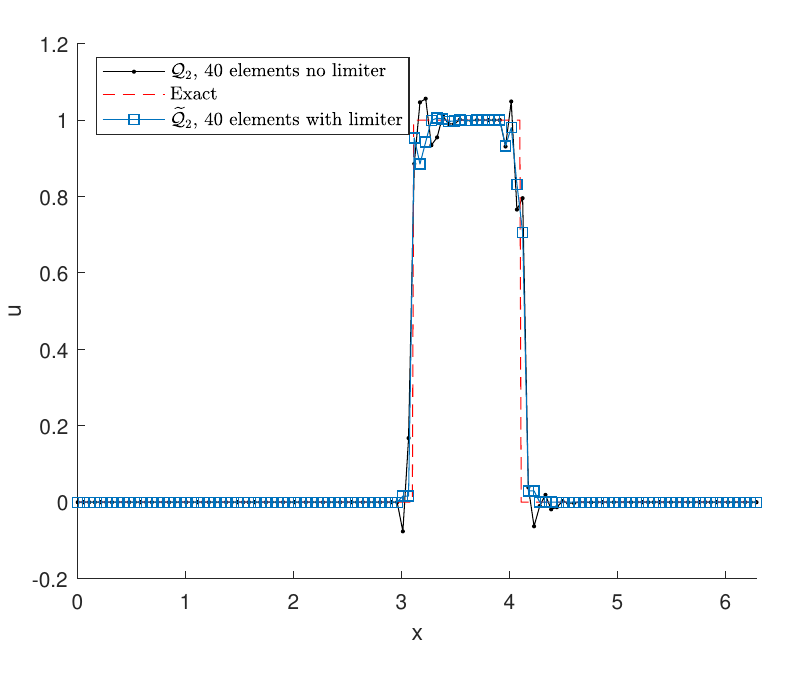}
\caption{$\mathcal{Q}_2$ and $\mathcal{\widetilde{Q}}_2$ DG schemes}
\end{subfigure}%
\\%
\begin{subfigure}[h]{0.49\linewidth}
\includegraphics[trim = 40mm 80mm 40mm 85mm, clip, width=\linewidth]{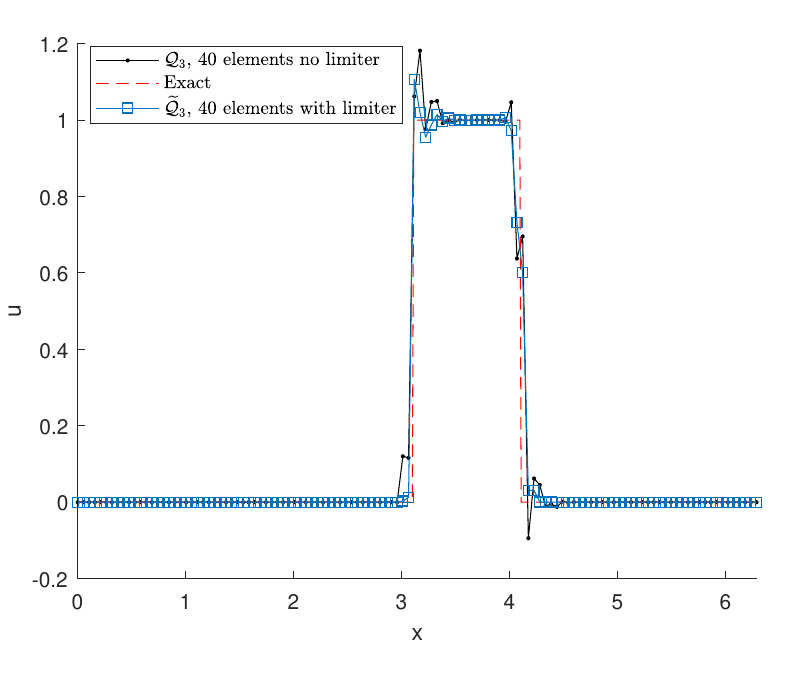}
\caption{$\mathcal{Q}_3$ and $\mathcal{\widetilde{Q}}_3$ DG schemes}
\end{subfigure}
\begin{subfigure}[h]{0.49\linewidth}
\includegraphics[trim = 40mm 80mm 40mm 85mm, clip, width=\linewidth]{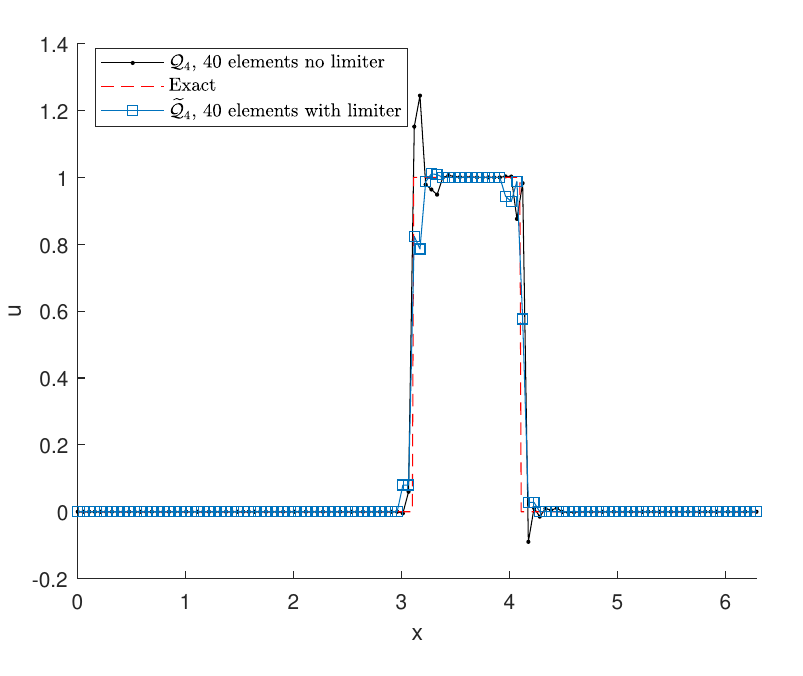}
\caption{$\mathcal{Q}_4$ and $\mathcal{\widetilde{Q}}_4$ DG schemes}
\end{subfigure}%
\caption{Comparison of the DG schemes for various polynomial orders, with and without the positivity-preserving limiter}
\label{fig:shock_ex1}
\end{figure}

\begin{figure}[htb!]
\centering
\begin{subfigure}[h]{0.49\linewidth}
\includegraphics[ width=\linewidth]{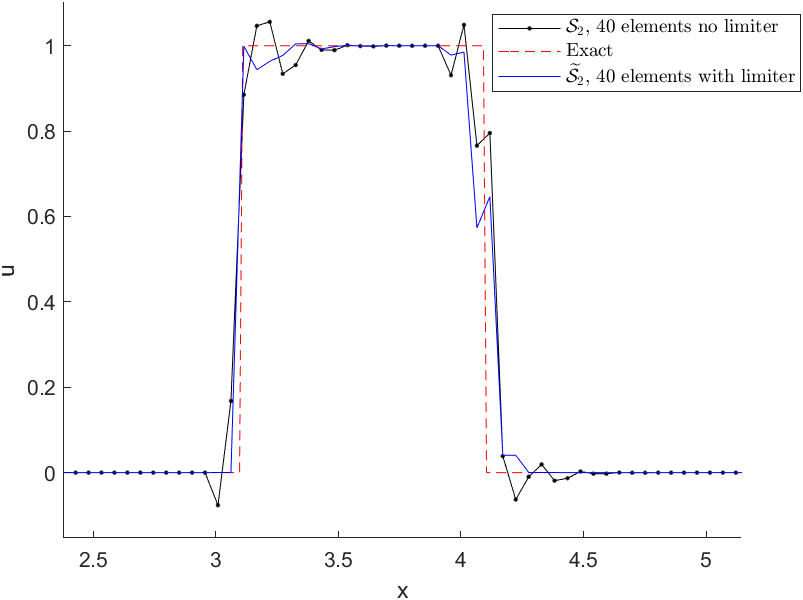}
\caption{$\mathcal{S}_2$ and $\mathcal{\widetilde{S}}_2$ DG schemes}
\end{subfigure}
\begin{subfigure}[h]{0.49\linewidth}
\includegraphics[  width=\linewidth]{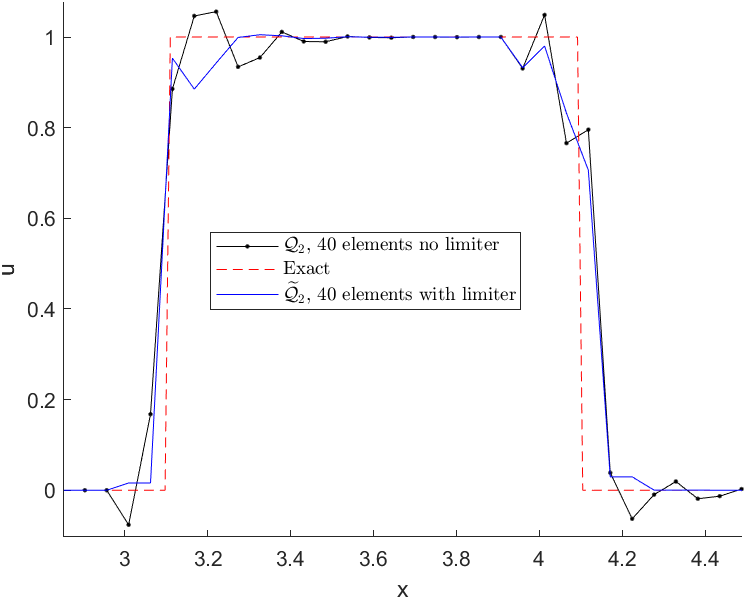}
\caption{$\mathcal{Q}_2$ and $\mathcal{\widetilde{Q}}_2$ DG schemes}
\end{subfigure}%
\\%
\begin{subfigure}[h]{0.49\linewidth}
\includegraphics[ width=\linewidth]{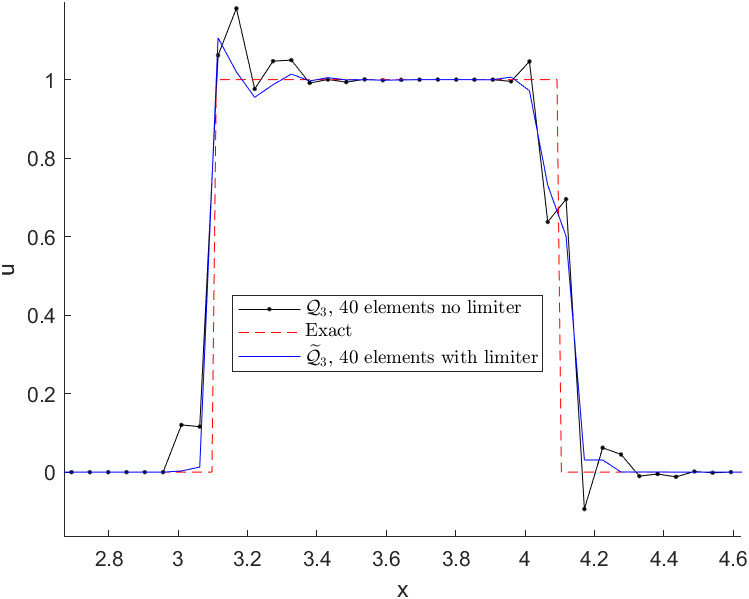}
\caption{$\mathcal{Q}_3$ and $\mathcal{\widetilde{Q}}_3$ DG schemes}
\end{subfigure}
\begin{subfigure}[h]{0.49\linewidth}
\includegraphics[ width=\linewidth]{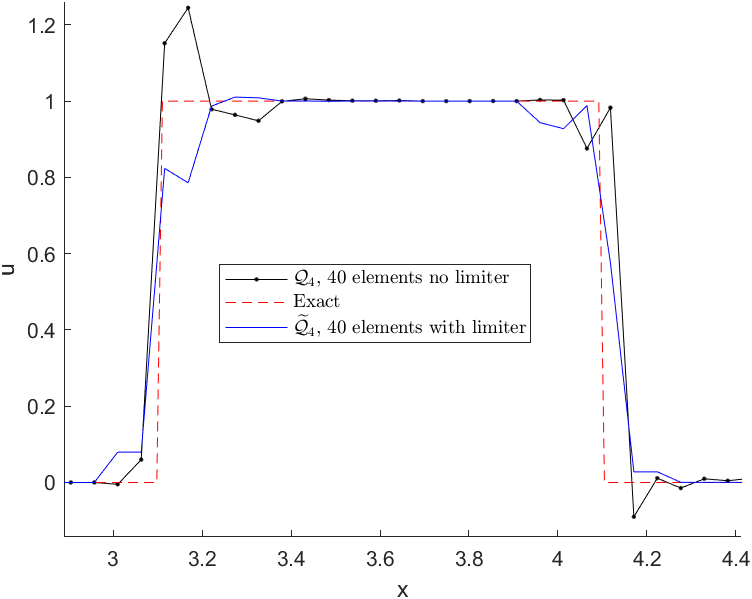}
\caption{$\mathcal{Q}_4$ and $\mathcal{\widetilde{Q}}_4$ DG schemes}
\end{subfigure}%
\caption{Zoom-in and comparison DG schemes for various polynomial orders, with and without the positivity-preserving limiter}
\label{fig:shock_ex1_zoom}
\end{figure}

\subsection{Accuracy test for the DG schemes in (1+2) dimensions} \label{section:numerics_accuracy_example3D}
A manufactured solution is used to verify numerically that the correct convergence rates are obtained.  We take $\alpha = \beta=8$ and $ \gamma=0.1$.  The computational domain is the cube $[0,0.5]^3.$  The exact solution is set as $u(x,y,t) = x y e^{t/8} e^{(x^2 + y^2)/8}.$ 

Table~\ref{table_example_section_manufactured3D} has the errors and convergence rates for the various 3D DG schemes.  The correct rate of $k+1$ in the $L^2$ norm is approximately observed for the unaugmented and augmented spaces.
\begin{table}[htb!]
  \centering 
  \begin{tabular}{lccccc}
    \toprule
     & &   \multicolumn{2}{c}{$\mathcal{X}_k$}  & \multicolumn{2}{c}{$ \mathcal{\widetilde{X}}_k$}\\
      \cmidrule(lr){3-4}\cmidrule(lr){5-6}
                              &$N_x=N_y=N_t$  & $L^2$ error  & $L^2$ order & $L^2$ error & $L^2$ order\\
    \midrule   
$\mathcal{X}_k=\mathcal{P}_1$ &$2$  &$1.266\mathrm{E-03}$    &$-$    &$1.283\mathrm{E-01}$ &$-$\\
				              &$4$  &$3.205\mathrm{E-04}$    &$1.98$ &$3.462\mathrm{E-02}$ & 1.88\\
							  &$8$  &$7.763\mathrm{E-05}$    &$2.04$ &$8.715\mathrm{E-03}$ & 1.99\\
				              &$16$ &$1.879\mathrm{E-05}$    &$2.04$ &$2.123\mathrm{E-03}$ & 2.03\\
				              \\
$\mathcal{X}_k=\mathcal{Q}_2$ &$2$  &$3.6674\mathrm{E-05}$    &$-$    &$3.654\mathrm{E-05}$ &$-$\\
				              &$4$  &$5.4880\mathrm{E-06}$    &$2.74$ &$5.472\mathrm{E-06}$ & 2.73\\
							  &$8$  &$5.9162\mathrm{E-07}$    &$3.21$ &$5.871\mathrm{E-07}$ & 3.22\\
				              &$16$ &$6.5933\mathrm{E-08}$    &$3.16$ &$6.532\mathrm{E-08}$ & 3.16\\ 
				              \\
$\mathcal{X}_k=\mathcal{Q}_3$ &$2$  &$4.285\mathrm{E-06}$    &$-$    &$4.285\mathrm{E-06}$ &$-$\\
				              &$4$  &$2.408\mathrm{E-07}$    &$4.14$ &$2.411\mathrm{E-07}$ & 4.15\\
							  &$8$  &$9.773\mathrm{E-09}$    &$4.62$ &$9.775\mathrm{E-09}$ & 4.62\\
				              &$16$ &$4.309\mathrm{E-10}$    &$4.50$ &$4.303\mathrm{E-10}$ & 4.50\\ 				              
    \bottomrule
  \end{tabular}
  \caption{Manufactured solution example in 3D~(Section~\ref{section:numerics_accuracy_example3D}).  Errors and rates for the 3D DG schemes with different spaces $\mathcal{X}_k$ and $ \mathcal{\widetilde{X}}_k$  }
  \label{table_example_section_manufactured3D}
\end{table}
 
\subsection{Step function initial condition for the linear problem in (1+2) dimensions}
Here we explore a natural extension of the step function propagation to two spatial dimensions. As before, we pick the time domain as $[0,1]$, which is partitioned into $N_t=16$ subintervals. The space domain is put as $[0,2\pi]^2$, which is partitioned into $N_x=N_y=32$ in each coordinate direction. For simplicity, we set $\alpha=\beta=1$, $f\equiv 0$, and $\gamma=0$. The following step function initial condition is utilized:
\[
u(x,y,0)
=
 \begin{cases} 
        1, & \text{for } (x,y)\in [2,3] \\
        0, & \text{otherwise} .
        \end{cases} 
\]
Homogeneous boundary conditions of the form $u(0,0,t)=0$ are assumed.  
 

The simple nature of unidirectional flow causes the initial profile to translate. It is evident that the approximation is not bound-preserving without augmentation and limiting (see Fig.~\ref{fig:shock_ex2_2}) 

On the other hand, when we utilize the augmentation and limiter, the corresponding $\mathcal{\widetilde{P}}_k$ DG schemes are positivity-preserving ($k=1,2,3,4$, Fig.~\ref{fig:shock_ex2}).

\begin{figure}[htb!]
\centering 
\includegraphics[scale = 0.4]{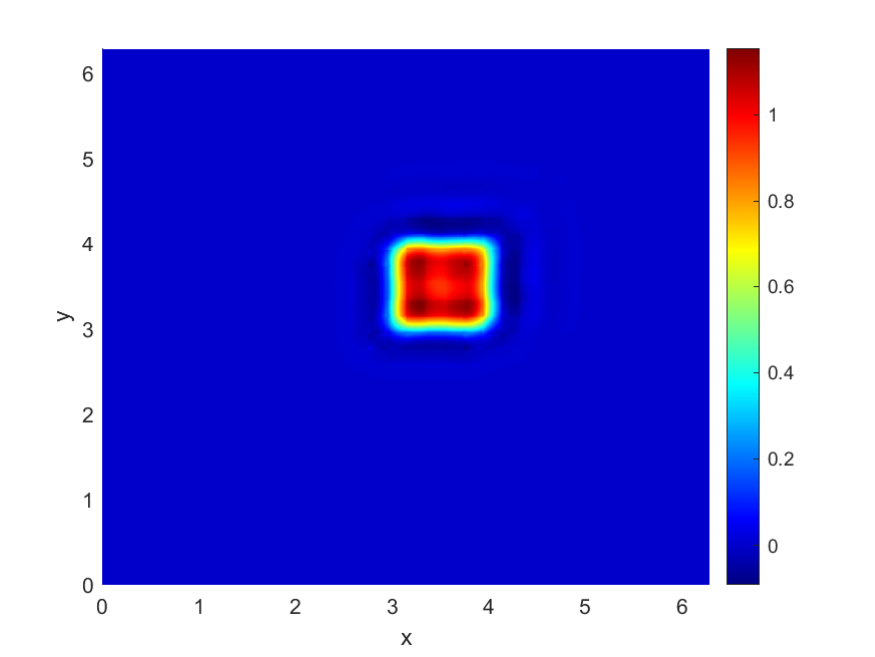} 
\caption{The $\mathcal{{P}}_2$ DG scheme without limiting. It is clear that there are violations in the bounds.}
\label{fig:shock_ex2_2}
\end{figure} 

\begin{figure}[htb!]
\centering
\begin{subfigure}[h]{0.49\linewidth}
\includegraphics[ scale = 0.4]{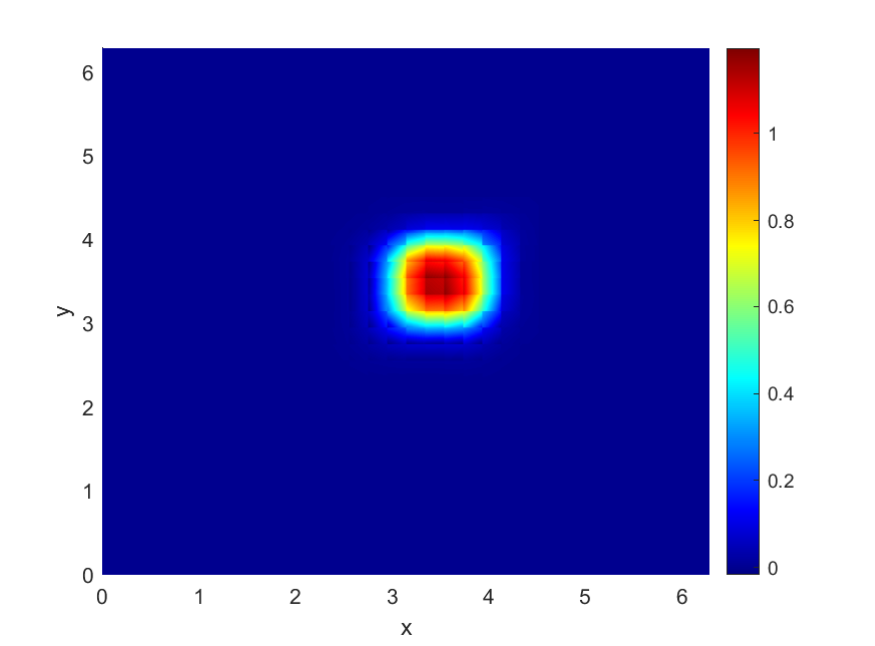}
\caption{$\mathcal{\widetilde{P}}_1$ DG method (limiting)}
\end{subfigure}
\begin{subfigure}[h]{0.49\linewidth}
\includegraphics[ scale = 0.4]{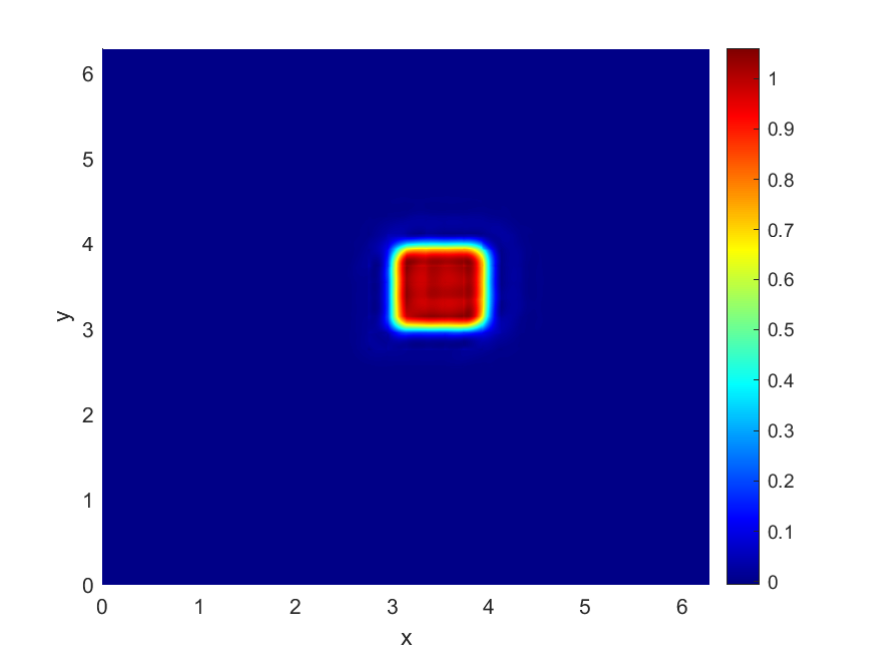}
\caption{$\mathcal{\widetilde{P}}_2$ DG method (limiting)}
\end{subfigure}%
\\%
\begin{subfigure}[h]{0.49\linewidth}
\includegraphics[ scale = 0.4]{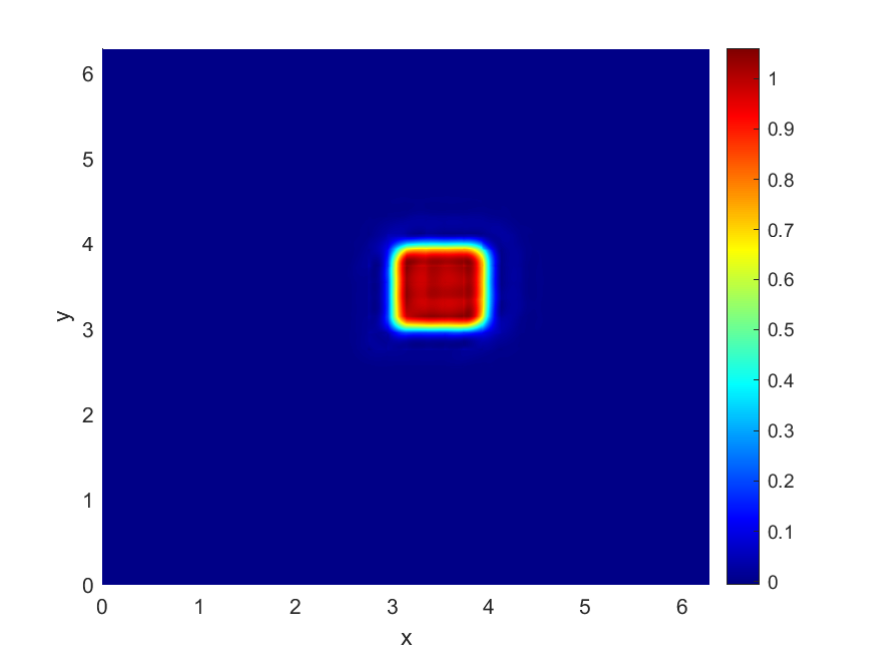}
\caption{$\mathcal{\widetilde{P}}_3$ DG method (limiting)}
\end{subfigure}
\begin{subfigure}[h]{0.49\linewidth}
\includegraphics[ scale = 0.4]{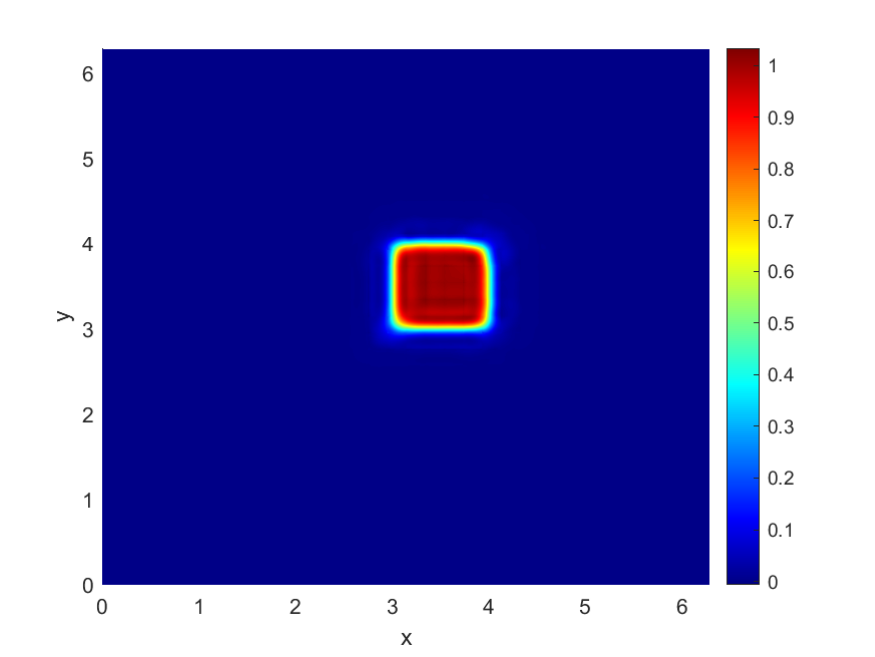}
\caption{$\mathcal{\widetilde{P}}_4$ DG method (limiting)}
\end{subfigure}%
\caption{
Results of the optimization based augmentation and simple scaling limiter, for $k=1,2,3,4$.
}
\label{fig:shock_ex2}
\end{figure}

\subsection{Computational performance}
The previous examples focus on verification and validation of the optimization-based enrichment technique. In this section, we focus on the computational performance. We make a few remarks before presenting the test problem and results.
 \begin{enumerate} 
 \item 
If the cell average from the unmodulated DG scheme using $\mathcal{ {X}}_k(S_{ij\ell})$ is negative, we apply Problem~\eqref{eq:optimization_algorithm3D} (or Problem~\eqref{eq:optimization_algorithm}). Otherwise, for all untroubled cells, there is no need to apply the optimization procedure.

\item 
The model PDE has unidirectional flow, which is known before the simulation starts (for example, the parameters $\alpha,\beta$ from \eqref{eq:model3D} are assumed to be given).  This means we can pre-process suitable augmented basis functions sufficient for the expected range of parameters. For instance, in Table~\ref{table_example_section_ex_S2}, several augmented basis functions have been pre-computed which are valid under a fixed CFL condition. 

\item   Since we know the CFL condition before the simulation begins, we do not require to run the optimization ``on the fly'' for each cell independently. Rather, similar to computing quadrature points on a reference element, we run numerically solve problem~\eqref{eq:optimization_algorithm3D} once, at the start of the simulation, for a given CFL condition. We then augment the standard DG spaces with this function on troubled cells.

\item We do not focus on enhancing the performance the method. For example, the code neither leverages parallelism nor exploits any tensor product structures. Default options are passed into the numerical interior point optimizer.  Additional attention to these details can further improve the performance.
\end{enumerate}
Keeping the above remarks in mind, we treat the construction of the augmented basis function as a pre-processing step - not too dissimilar from mesh generation or quadrature construction on reference cells.
 
In all computational experiments, we solve the optimization problems (Problems~\eqref{eq:optimization_algorithm3D} and ~\eqref{eq:optimization_algorithm}) approximately via interior point methods for nonlinear programming \cite{byrd2000trust}. We focus on the reference cell $[-1,1]^2$ in 2D and $[-1,1]^3$ in 3D, as is standard in many finite element computations \cite{quarteroni2001numerical}.  The parameter space $\Omega_\delta = (\alpha,\beta,\omega,\gamma,\Delta x, \Delta y,\Delta z)$ is sampled so that the augmented basis functions are valid for a wide range of CFL numbers (see Fig.~\ref{fig:cfl_ex}).  
\begin{table}[htb!]
\centering
\begin{tabular}{|c|c|c|}
\hline
$k$ & \(\widetilde{\mathcal{P}}_k\)   & \(\widetilde{\mathcal{Q}}_k\)  \\
\hline
1 & \(7.5836E-01\) & \(9.6983E-01\) \\
2 & \(9.9454E-01\) & \(2.7070E+00\) \\
3 & \(1.8614E+00\) & \(7.2597E+00\) \\
4 & \(4.3692E+00\) & \(1.5475E+01\) \\
5 & \(1.0319E+01\) & \(3.9193E+01\) \\
6 & \(1.9814E+01\) & \(1.1344E+02\) \\
7 & \(5.3975E+01\) & \(3.1415E+02\) \\
8 & \(8.3156E+01\) & \(5.2043E+02\) \\
\hline
\end{tabular}
\caption{Computation times in seconds for \( \widetilde{\mathcal{P}}_k \) and \( \widetilde{\mathcal{Q}}_k \) in 2D for various polynomial degrees. Here, the augmented functions are enforced to be valid over a large parameter space $\Omega_\delta = (\alpha,\beta,\gamma,\Delta x, \Delta y)$.}
\label{table_comp2}
\end{table}
The computational results in 2D can be found in Table~\ref{table_comp2}.  We see that as the polynomial degree $k$ increases, the time to find a suitable basis function also increases. It should be emphasized that these timings are done once at the start of the simulation, and are of little overhead (compared to mesh generation, for instance).  The augmented tensor product space $\widetilde{\mathcal{Q}}_k$ takes longer to compute than \(\widetilde{\mathcal{P}}_k\) since it has more degrees of freedom (and thus, more parameters to search for, see the vector $\vec{d}$ in Problem~\eqref{eq:optimization_algorithm}). 
\begin{table}[htb!]
\centering
\begin{tabular}{|c|c|c|}
\hline
$k$ & \(\widetilde{\mathcal{P}}_k\)   & \(\widetilde{\mathcal{Q}}_k\)  \\
\hline
1&     $6.524E-04$ &$8.6006E-04$\\
2&     $3.584E-03$ &$5.8656E-03$\\
3&     $4.441E-03$ &$8.6818E-03$\\
4&     $5.224E-03$ &$1.1875E-02$\\
5&     $8.226E-03$ &$2.1318E-02$\\
6&     $9.272E-03$ &$2.6980E-02$\\
7&     $1.961E-02$ &$6.3304E-02$\\
8&     $9.273E-02$ &$3.2886E-01$\\
\hline
\end{tabular}
\caption{Computation times in seconds for \( \widetilde{\mathcal{P}}_k \) and \( \widetilde{\mathcal{Q}}_k \) in 2D for various polynomial degrees. Here, the augmented functions are enforced to be valid over a \textit{fixed} parameter space $\Omega_\delta = (\alpha,\beta,\gamma,\Delta x, \Delta y)$.  For simplicity, $\alpha=\beta=\gamma=1$, $\Delta x= \Delta y = 2$.}
\label{table_comp3}
\end{table} 
On the other hand, if we use Problem~\eqref{eq:optimization_algorithm}) ``on the fly'', meaning we invoke the optimization problem on each troubled cell, then the timings are significantly reduced - because the augmented function needs to only hold for a \textit{fixed} parameter space $\Omega_\delta = (\alpha,\beta,\omega,\gamma,\Delta x, \Delta y,\Delta z)$. Table~\ref{table_comp3} displays the timings in this case.
 
In Table~\ref{table_comp4}, we showcase the timings in 3D. It is evident that the optimization is more expensive. For larger polynomial degrees, the algorithm starts to stagnate more, due to poor initial guesses for the nonlinear optimizer. Also, each iteration is of course more expensive in 3D, which contributes to the poor scaling.
\begin{table}[h]
\centering
\begin{tabular}{|c|c|c|}
\hline
$k$ & \(\widetilde{\mathcal{P}}_k\)  & \(\widetilde{\mathcal{Q}}_k\) \\
\hline
1 & \(1.5167E+00\)  & \(1.9397E+00\)  \\
2 & \(2.9836E+00\)  & \(8.1212E+00\)  \\
3 & \(7.4456E+00\)  & \(2.9039E+01\)  \\
4 & \(2.1846E+01\)  & \(7.7374E+01\)  \\
5 & \(6.1914E+01\)  & \(2.3516E+02\)  \\
6 & \(1.3869E+02\)  & \(5.7435E+02\)  \\
7 & \(4.3180E+02\)  & \(7.1354E+02\)  \\
8 & \(8.4840E+02\)  & \(1.5274E+03\)  \\
\hline
\end{tabular}
\caption{Computation times in seconds for \( \widetilde{\mathcal{P}}_k \) and \( \widetilde{\mathcal{Q}}_k \) in 2D for various polynomial degrees. Here, the augmented functions are enforced to be valid over a large parameter space $\Omega_\delta = (\alpha,\beta,\omega,\gamma,\Delta x, \Delta y,\Delta z)$.}
\label{table_comp4}
\end{table}

\section{Conclusion}
	A high-order positivity-preserving DG scheme is presented in this paper.  In~\cite{ling2018conservative}, a second order accurate positivity-preserving DG scheme is studied.  In the two dimensional case, they augment the traditional DG space $\mathcal{P}_1$ with an additional basis function, so that the cell average of the unmodulated DG space remains non-negative.  Our method in this paper extends the results found in~\cite{ling2018conservative} to polynomial degrees greater than one.  To do so, a similar premise is assumed.  A standard space such $\mathcal{X}_k$ ($\mathcal{P}_k$, $\mathcal{Q}_k$, or $\mathcal{S}_k$) is augmented with an anzats function $\psi\in \mathcal{X}_r$ for some $r>k$.  Then, we numerically search for a suitable function $\psi$ through a constrained nonlinear optimization problem.  When $\psi$ is augmented to the standard DG space $\mathcal{X}_k$, it will preserve a positive cell average for the unmodulated DG solution.  As the resulting cell average is guaranteed to be positive, the entire solution can be made positive by applying a simple scaling limiter~\cite{zhang2010maximum}.  This limiter scales the solution towards the cell average without reducing its formal accuracy~\cite{zhang2010maximum}.  
	
	We observe numerically that the method has good performance for any CFL condition (arbitrarily large or small).  It is also feasible to obtain a single augmented basis function that is valid for a wide range of CFL conditions; several explicit basis functions were presented.  The augmented DG space involves additional computational cost.  Some of this cost may be mitigated by using an adaptive scheme, where $\mathcal{X}_k$ is used when the cell average is positive, and the augmented space is only used when the cell average from $\mathcal{X}_k$ is negative. Moreover, this technique may be used to preprocess suitable augmented basis functions, as opposed to using Problem~\eqref{eq:optimization_algorithm} on the fly.
	
	Numerical experiments are conducted in 2D and 3D, and they demonstrate that the method works as expected for polynomial degrees $k>1$.  Several standard DG spaces are also explored, such as $\mathcal{P}_k$, $\mathcal{Q}_k$, and the serendipity spaces $\mathcal{S}_k$.  Additionally, careful convergence studies are provided. 
	
	In future work we plan to examine extension of the presented scheme for nonlinear problems, applications such as radiative transport, and unstructured meshes.

\noindent \textbf{Funding} This study was funded by the University of Wisconsin-Madison.\newline
\noindent \textbf{Conflicts of interest/Competing interests} All authors have no conflicts of interest.\newline
\noindent \textbf{Data Availability} Enquirers about data availability should be directed to the authors\newline
\section*{Declarations}
\textbf{Conflict of interest} The author has no competing interests to declare that are relevant to the content of this article.
   
\section*{Appendix A} \label{section:Q2_proof}
Below (see~\eqref{eq_app_A_Q2_basis}) are the tensor product Bernstein polynomials used as a basis for $\mathcal{Q}_2$.
\begin{subequations}
\begin{align} 
\phi_1 (\xi,\eta) &= (\eta/2 - 1/2)^2 (\xi/2 - 1/2)^2
\\
\phi_2 (\xi,\eta) &=-(\eta/2 - 1/2) (\xi/2 - 1/2)^2 (\eta + 1)
 \\
\phi_3 (\xi,\eta) &=(\eta/2 + 1/2)^2 (\xi/2 - 1/2)^2
 \\
\phi_4 (\xi,\eta) &=-(\eta/2 - 1/2)^2 (\xi/2 - 1/2) (\xi + 1)
 \\
\phi_5 (\xi,\eta) &=(\eta/2 - 1/2)(\xi/2 - 1/2)(\eta + 1)(\xi + 1)
  \\
\phi_6 (\xi,\eta) &=-(\eta/2 + 1/2)^2(\xi/2 - 1/2)(\xi + 1)
 \\
\phi_7 (\xi,\eta) &=(\eta/2 - 1/2)^2(\xi/2 + 1/2)^2
\\
\phi_8 (\xi,\eta) &= -(\eta/2 - 1/2)(\xi/2 + 1/2)^2(\eta + 1)
  \\
\phi_9 (\xi,\eta) &=(\eta/2 + 1/2)^2(\xi/2 + 1/2)^2
\end{align}
\label{eq_app_A_Q2_basis}
\end{subequations} 
\subsection*{Case for $B< 0.5$, $\gamma = 0$}
The purpose of this section is to provide analytic evidence that $ \{\psi_{1,\mathcal{Q}_2}\}$ (see \ref{table_example_section_ex_Q2}) augmented to $\mathcal{Q}_2$ ensures a non-negative cell average provided that $B<1/2$ and $\gamma=0$.  Recall $\hat{S}=[-1,1]^2$ is the reference cell. Write $\widetilde{\mathcal{Q}}_2 = \mathcal{Q}_2 \cup \text{span}\{\psi_{1,\mathcal{Q}_2}\}$.  Express the solution to the linear system arising from
\begin{equation} 
\mathcal{L}(u,v) = \iint_{\widehat{S}} u\,d\xi d\eta ,~~~\forall u \in \widetilde{\mathcal{Q}}_2(\widehat{S}),
\label{eq_anchor}
\end{equation}
as
\[
v(\xi,\eta) = \sum_{i=1}^{9} c_i \phi_i(\xi,\eta) + c_{10}\psi_{1,\mathcal{Q}_2}(\xi,\eta),
\] 
where $\psi_{1,\mathcal{Q}_2}$ is the augmented basis function. With the choice $\psi_{1,\mathcal{Q}_2}= (\xi \eta^2(1-\xi)^2(\xi + 1)(1-\eta^2)/8)^2$, and $\phi_i$ are defined in~\eqref{eq_app_A_Q2_basis}, solving the linear system arising from~\eqref{eq_anchor} generates the following coefficients:
\begin{align*}
c_1 &= \frac{\Delta x}{(B+1)\Lambda} (2020 B^7 + 5718 B^6 + 28920 B^5 + 29148 B^4 + 60834 B^3 + 42181 B^2 + 9265 B + 2350)\\
c_2 &= \frac{\Delta x}{(B+1)\Lambda} (1510 B^7 + 4364 B^6 + 11517 B^5 + 12607 B^4 + 30325 B^3 + 42821 B^2 + 23365 B + 2350)\\
c_3 &= \frac{\Delta x}{(B+1)\Lambda} (               - 190 B^6 + 1944 B^5 - 2174 B^4 - 9774 B^3 - 5889 B^2 + 2215 B + 2350)\\
c_4 &= \frac{\Delta x}{(B+1)\Lambda} (  2020 B^7 + 20838 B^6 + 37248 B^5 + 38336 B^4 + 29439 B^3 + 2828 B^2 - 955 B + 1175)\\
c_5 &= \frac{\Delta x}{(B+1)\Lambda} ( 1110 B^7 + 5434 B^6 + 22499 B^5 + 53315 B^4 + 46611 B^3 + 26933 B^2 + 9545 B + 1175)\\
c_6 &= \frac{\Delta x}{(B+1)\Lambda} (                  400 B^6 - 3180 B^5 - 2241 B^4 + 7958 B^3 + 6013 B^2 - 955 B + 1175)\\
c_7 &= \frac{\Delta x}{(B+1)\Lambda} (             2020 B^7 + 658 B^6 - 7744 B^5 - 8846 B^4 - 4256 B^3 + 1595 B^2 + 3675 B)\\
c_8 &= \frac{\Delta x}{(B+1)\Lambda} (              960 B^7 + 5004 B^6 + 5566 B^5 - 5072 B^4 - 5728 B^3 + 815 B^2 - 1575 B)\\
c_9 &= \frac{\Delta x}{(B+1)\Lambda} (                       B(- 110 B^5 + 1556 B^4 + 3632 B^3 + 1640 B^2 + 835 B + 3675))\\
c_{10} &= \frac{\Delta x}{4\Lambda} (                        -121275 B(50 B^5 - 70 B^4 + 639 B^3 + 317 B^2 + 95 B - 525))
\\
\Lambda &= (2020 B^7 + 8228 B^6 + 14544 B^5 + 19129 B^4 + 22737 B^3 + 19926 B^2 + 10440 B + 2350)
\end{align*}
Using the assumption $B\le 0.5$, we can show $c_i\ge0$ for $i\neq 8$.  Note that $\Lambda>0$, as it is a polynomial with positive coefficients. Therefore, the denominators of $c_i$ are  positive. To inspect the sign of $c_i$, it suffices to examine the univariate polynomial terms (functions of $B$) in the numerator. Define these polynomials as follows: $p_i(B) = ((B+1)\Lambda/\Delta x) c_i$ for $1\le i \le 9$ and $p_{10}(B) =   (4\Lambda/\Delta x) c_{10}$.

It is straightforward to verify that $p_1, p_2, p_5$ are the sum of non-negative terms. The following lower bounds are immediate:
\begin{align*}
p_4&\ge - 955 B + 1175\ge 0,\\
p_6&\ge  - 3180 B^5 - 2241 B^4 - 955 B + 1175\ge0,\\
p_9&\ge B(3675 - 110B^5) \ge0.
\end{align*}
The remaining polynomials $p_3, p_7, p_8$, and $p_{10}$ are also non-negative (see Figure~\ref{eq_anchor}). Thus, $c_8$ is the only coefficient that is non-positive.
\begin{figure}[htb!]
    \centering
    \begin{subfigure}[b]{0.35\textwidth}
        \centering
        \includegraphics[width=\textwidth]{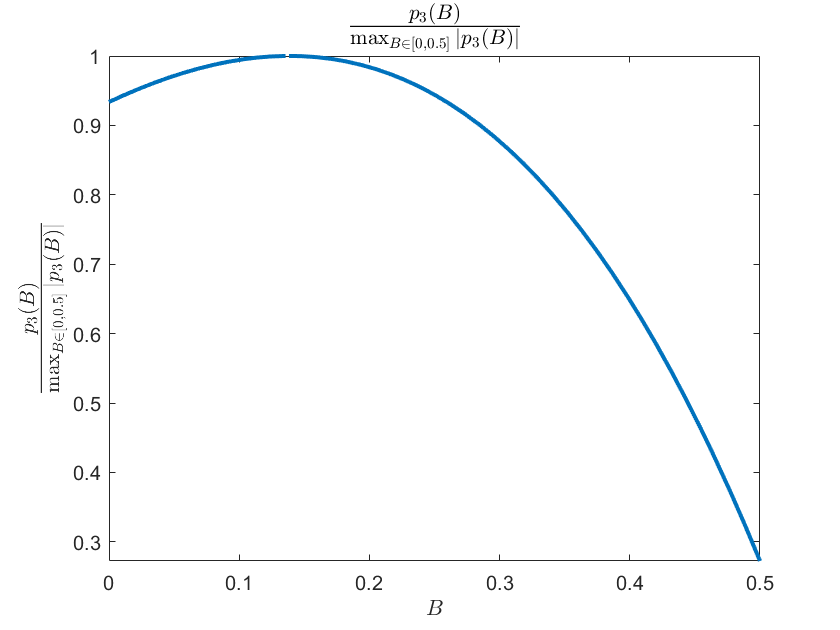}
        \caption[]{{\small $p_3$ normalized}}
        \label{fig:Q2_coeffs_c3}
    \end{subfigure}
    \begin{subfigure}[b]{0.35\textwidth}
        \centering
        \includegraphics[width=\textwidth]{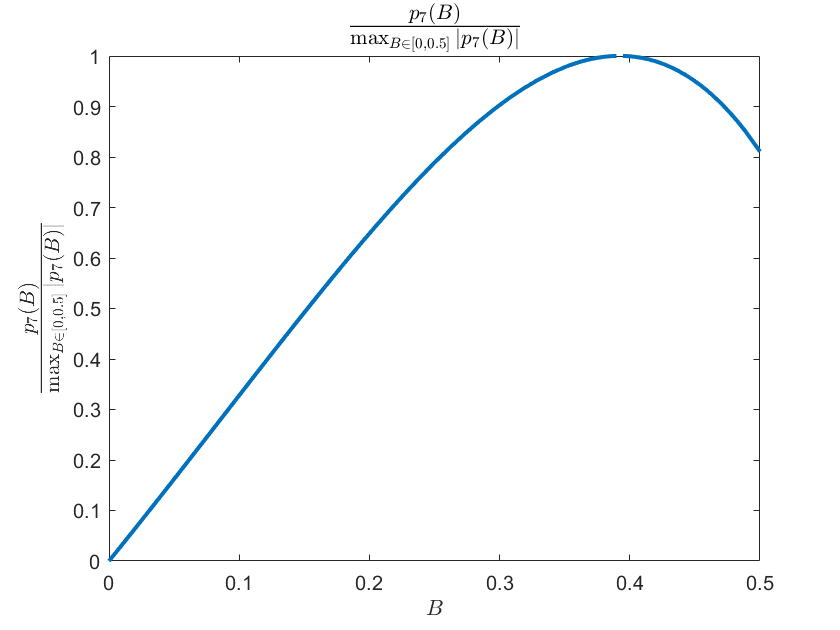}
        \caption[]{{\small $p_7$ normalized}}
        \label{fig:Q2_coeffs_c7}
    \end{subfigure}
\vskip\baselineskip
    \begin{subfigure}[b]{0.35\textwidth}
        \centering
        \includegraphics[width=\textwidth]{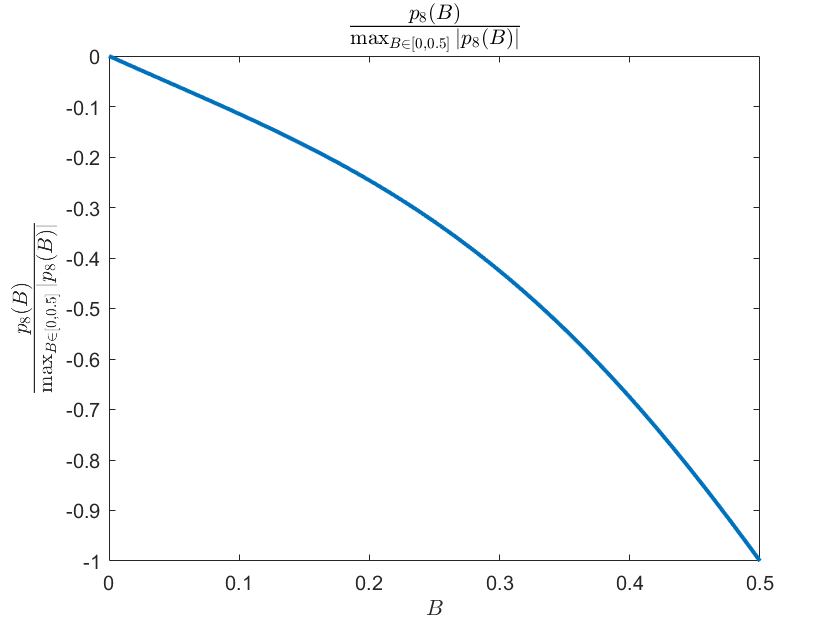}
        \caption[]{{\small $p_8$ normalized}}
        \label{fig:Q2_coeffs_c8}
    \end{subfigure}
    \begin{subfigure}[b]{0.35\textwidth}
        \centering
        \includegraphics[width=\textwidth]{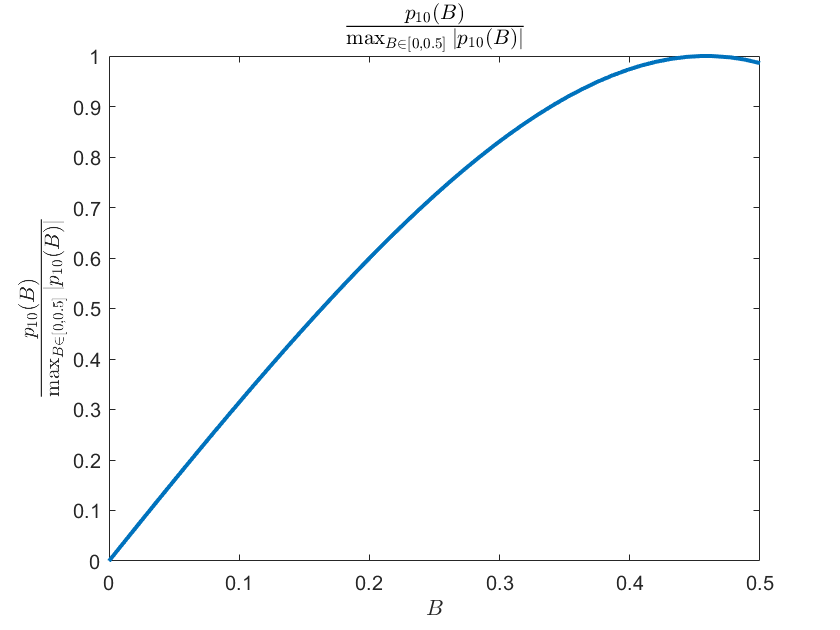}
        \caption[]{{\small $p_{10}$ normalized}}
        \label{fig:Q2_coeffs_c10}
    \end{subfigure}
\caption{Case of $B\le 0.5$, $G=\gamma =0$. Plots of $\frac{p_i(B)}{\max_{B\in [0,0.5]} |p_i(B)|}$ for $i=3,7,8,10$.}
\label{fig:Q2_coeffs}
\end{figure} 
With this result, we have
\begin{align} 
v(\xi,\eta) &= \sum_{i=1}^{9} c_i \phi_i(\xi,\eta) + c_{10}\psi_{1,\mathcal{Q}_2}(\xi,\eta) \notag
\\
	   &\ge  c_9 \phi_9(\xi,\eta) + c_8 \phi_8(\xi,\eta) \notag
	   \\
	   &= \frac{1}{4}(\xi+1)^2  \bigg( c_9 \frac{(1+\eta)(1+\eta)}{4} + c_8 \frac{(1+\eta)(1 - \eta)}{2} \bigg) \notag
	   \\
	   &= \frac{1}{4}(\xi+1)^2  w(\eta) , \label{eq:Q2_inequality}
\end{align}
where $w(y) = c_9 \frac{(1+\eta)(1+\eta)}{4} + c_8 \frac{(1+\eta)(1 - \eta)}{2}$, and we used $c_i\ge 0$ for $i\neq 8$ and the non-negativity of the basis functions. The quadratic function $w$ is concave because $w''(\eta) = \frac{c_9}{2} - c_8\ge 0$. Moreover, $w(\pm 1)\ge 0$, therefore, $w(\eta) \ge0$. From~\eqref{eq:Q2_inequality}, it is evident that $v(\xi,\eta)\ge 0$.       

\section*{Appendix B: $\mathcal{Q}_3 $ and $\mathcal{Q}_4$ augmentation}\label{section:Qk_example} 
This appendix contains examples of augmented basis functions for $\mathcal{Q}_k$ generated from Problem~\eqref{eq:optimization_algorithm}. Table \ref{table_example_section_ex_Q2} is duplicated here for convenience.
\begin{table}[htb!]
  \centering 
  \begin{tabular}{clc}
    \toprule
 Condition to enforce $\overline{u}_{ij}\ge0$ & Augmented basis function $\psi$                        & Original space \\
    \midrule  
  $(\beta/\alpha)\Delta x/\Delta y < \frac{1}{2}$ & $\psi_{1,\mathcal{Q}_2} =  (xy^2(1-x)^2(x + 1)(1-y^2)/8)^2 $ & $\mathcal{Q}_2$   \\
  $(\beta/\alpha)\Delta x/\Delta y > 2$           & $\psi_{2,\mathcal{Q}_2} = (y x^2 (1-y)^2 (y + 1)   (1-x^2)/8)^2$ & $\mathcal{Q}_2$\\  
  $(\beta/\alpha)\Delta x/\Delta y < \frac{1}{4}$ & $\psi_{1,\mathcal{Q}_3}  $ (see \eqref{eq:Appendix_Q3_psi_1})
  & $\mathcal{Q}_3$   \\ 
  $(\beta/\alpha)\Delta x/\Delta y > 4$           & $\psi_{2,\mathcal{Q}_3}  $ (see \eqref{eq:Appendix_Q3_psi_2}) & $\mathcal{Q}_3$\\   
  
  $(\beta/\alpha)\Delta x/\Delta y < \frac{1}{8}$ & $\psi_{1,\mathcal{Q}_4}  $ (see \eqref{eq:Appendix_Q4_psi_1}) & $\mathcal{Q}_4$   \\ 
  $(\beta/\alpha)\Delta x/\Delta y > 8$           & $\psi_{2,\mathcal{Q}_4}  $ (see \eqref{eq:Appendix_Q4_psi_2}) & $\mathcal{Q}_4$\\    
    \bottomrule
  \end{tabular}
  \caption{Example augmented basis functions for $\mathcal{Q}_k$ generated by Problem~\ref{eq:optimization_algorithm}.  
  }
\end{table}
{{
\setlength{\abovedisplayskip}{10pt}
\setlength{\belowdisplayskip}{10pt} 
\begin{align}
r_1(\xi,\eta) &=
 (\xi - \sqrt{5}/5)  (\eta - \sqrt{5}/5)  (\eta + \sqrt{5}/5)  (\xi - 1)  (\xi + 1)  (\eta - 1)
\\
r_2(\xi,\eta) &=
 (\xi + \sqrt{5}/5)  (\eta - \sqrt{5}/5)  (\eta + \sqrt{5}/5)  (\xi - 1)  (\xi + 1)  (\eta + 1)
\\
r_3(\xi,\eta) &=   (\xi - \sqrt{5}/5)  (\xi + \sqrt{5}/5)  (\eta - \sqrt{5}/5)  (\xi - 1)  (\eta - 1)  (\eta + 1)
\\
r_4(\xi,\eta) &=
   (\xi - \sqrt{5}/5)  (\xi + \sqrt{5}/5)  (\eta + \sqrt{5}/5)  (\xi + 1)  (\eta - 1)  (\eta + 1)
\\
\psi_{1,\mathcal{Q}_3}(\xi,\eta) &= ((25  \sqrt{5} )^2r_1(\xi,\eta) r_2(\xi,\eta)/8^4)^2 
\label{eq:Appendix_Q3_psi_1}
\\
\psi_{2,\mathcal{Q}_3}(\xi,\eta) &= ((25  \sqrt{5} )^2r_3(\xi,\eta) r_4(\xi,\eta)/8^4)^2
\label{eq:Appendix_Q3_psi_2}
\end{align}
Below is the augmented function ($\psi_{1,\mathcal{Q}_4}$, \eqref{eq:Appendix_Q4_psi_1})) for $\mathcal{Q}_4$ with $B<1/8$. 
\begin{align}
w_1(\xi) &=\xi  (\xi - 1) (\xi - (\sqrt{3}  \sqrt{7})/7) (\xi + (\sqrt{3}  \sqrt{7})/7)
\\
w_2(\eta) &=  \eta   (\eta - 1)       (\eta - (\sqrt{3}  \sqrt{7})/7)  (\eta + (\sqrt{3}  \sqrt{7})/7)
\\
w_3(\xi ) &= \xi (\xi - 1)  (\xi + 1)(\xi - (\sqrt{3}  \sqrt{7})/7) 
\\
w_4(\eta) &= \eta    (\eta + 1)    (\eta - (\sqrt{3}  \sqrt{7})/7)  (\eta + (\sqrt{3}  \sqrt{7})/7)
\\
r_1(\xi,\eta) &=   
(49 w_1(\xi) w_2(\eta)  )/64
 \\
r_2(\xi,\eta) &=
-(343   w_3(\xi )  w_4(\eta) )/192
\\
\psi_{1,\mathcal{Q}_4}(\xi,\eta) &= (r_1(\xi,\eta)r_2(\xi,\eta))^2
\label{eq:Appendix_Q4_psi_1}
\end{align} 
Below is the augmented function ($\psi_{2,\mathcal{Q}_4}$, \eqref{eq:Appendix_Q4_psi_2}) for $\mathcal{Q}_4$ with $B>8$.
\begin{align}
w_1(\xi)&=\xi(\xi - 1) (\xi - (\sqrt{3}  \sqrt{7})/7)  (\xi + (\sqrt{3}  \sqrt{7})/7)
\\
w_2(\eta)&=\eta    (\eta - 1)     (\eta - (\sqrt{3}  \sqrt{7})/7)  (\eta + (\sqrt{3}  \sqrt{7})/7)
\\
w_3(\xi)&=\xi(\xi + 1)(\xi - (\sqrt{3}  \sqrt{7})/7)  (\xi + (\sqrt{3}  \sqrt{7})/7) 
\\
w_4(\eta)&=\eta     (\eta - 1)  (\eta + 1)    (\eta - (\sqrt{3}  \sqrt{7})/7)
\\
r_1(\xi,\eta) &=
(49  w_1(\xi)w_2(\eta)    )/64
 \\
r_2(\xi,\eta) &=
-(343  w_3(\xi)w_4(\eta)    )/192
\\
\psi_{2,\mathcal{Q}_4}(\xi,\eta) &= (r_1(\xi,\eta)r_2(\xi,\eta))^2
\label{eq:Appendix_Q4_psi_2}
\end{align} 
}}

\section*{Appendix C: Augmentation for $\mathcal{S}_3$ and $\mathcal{S}_4$} \label{section:serendipity_example}
This appendix contains examples of augmented basis functions for $\mathcal{S}_k$ generated from Problem~\eqref{eq:optimization_algorithm}. Table \ref{table_example_section_ex_S2} is included here for convenience.
\begin{table}[htb!]
  \centering 
  \begin{tabular}{clc}
    \toprule
 Condition to enforce $\overline{u}_{ij}\ge0$ & Augmented basis function $\psi$                        & Original space \\
    \midrule  
  $(\beta/\alpha)\Delta x/\Delta y < \frac{1}{2}$ & $\psi_{1,\mathcal{S}_2} = 
  (
\xi  \eta^2  (\xi - 1)^2  (\eta - 1)    (\xi + 1)  (\eta + 1)/8
  )^2$ 
  & $\mathcal{S}_2$   \\
  $(\beta/\alpha)\Delta x/\Delta y > 2$           & 
  $\psi_{2,\mathcal{S}_2} 
  = (
\xi^2  \eta  (\xi - 1)  (\eta - 1)^2 
  (\xi + 1)  (\eta + 1)/2
  )^2$
   & $\mathcal{S}_2$\\     
  $(\beta/\alpha)\Delta x/\Delta y <\frac{1}{8}$ & $\psi_{1,\mathcal{S}_3} $ (see \eqref{eq:Appendix_S3_psi_1}) & $\mathcal{S}_3$   \\ 
  $(\beta/\alpha)\Delta x/\Delta y > 8$           & $\psi_{2,\mathcal{S}_3}  $ (see \eqref{eq:Appendix_S3_psi_2}) & $\mathcal{S}_3$\\   
  
  $(\beta/\alpha)\Delta x/\Delta y < \frac{1}{16}$ & $\psi_{1,\mathcal{S}_4} $ (see \eqref{eq:Appendix_S4_psi_1}) & $\mathcal{S}_4$   \\ 
  $(\beta/\alpha)\Delta x/\Delta y > 16$           & $\psi_{2,\mathcal{S}_4}  $ (see \eqref{eq:Appendix_S4_psi_2}) & $\mathcal{S}_4$\\    
    \bottomrule
  \end{tabular}
  \caption{Example augmented basis functions for the serendipity spaces $\mathcal{S}_k$ generated by Problem~\ref{eq:optimization_algorithm}. 
  } 
\end{table}
\begin{align}
r_1(\xi,\eta) &= 
  (\xi - \sqrt{5}/5)  (\eta - \sqrt{5}/5)^2  (\xi - 1)^2  (\xi + 1)^2  (\eta - 1)  (\eta + 1)^2  
\\
r_2(\xi,\eta) &=
   \sqrt{5}  (\xi + \sqrt{5}/5)  (\eta + \sqrt{5}/5)   
\\
r_3(\xi,\eta)&=
 (\xi - \sqrt{5}/5)^2  (\eta - \sqrt{5}/5)  (\xi - 1)  (\xi + 1)^2  (\eta - 1)^2  (\eta + 1)^2  
\\
r_4(\xi,\eta)&=
 \sqrt{5}   (\xi + \sqrt{5}/5)  (\eta + \sqrt{5}/5)   
 \\
\psi_{1,\mathcal{S}_3}(\xi,\eta) &= ((5^5/64^2) r_1(\xi,\eta)r_2(\xi,\eta) )^2
\label{eq:Appendix_S3_psi_1}
\\
\psi_{2,\mathcal{S}_3}(\xi,\eta) &= ((5^5/64^2) r_3(\xi,\eta)r_4(\xi,\eta) )^2
\label{eq:Appendix_S3_psi_2}
\end{align}       
\begin{align} 
r_1(\xi,\eta) &= 
  \xi^2 \eta^2 (\xi - 1)^2 (\eta - 1) (\xi - (\sqrt{3} \sqrt{7})/7) (\xi + (\sqrt{3} \sqrt{7})/7)^2 
\\
r_2(\xi,\eta) &=
    (\xi + 1) (\eta + 1) 
 (\eta - (\sqrt{3} \sqrt{7})/7)^2 (\eta + (\sqrt{3} \sqrt{7})/7)^2 
\\ 
r_3(\xi,\eta) &= 
  \xi^2 \eta (\xi - 1) (\xi + 1)^2 (\eta - 1)^2 (\xi - (\sqrt{3} \sqrt{7})/7)^2  
\\
r_4(\xi,\eta) &= 
 (\eta + 1) (\xi + (\sqrt{3} \sqrt{7})/7) (\eta - (\sqrt{3} \sqrt{7})/7)^2 (\eta + (\sqrt{3} \sqrt{7})/7)^2 
 \\
\psi_{1,\mathcal{S}_4}(\xi,\eta) &= ( (343/192) (49/64) r_1(\xi,\eta)r_2(\xi,\eta) )^2   
\label{eq:Appendix_S4_psi_1}
 \\
\psi_{2,\mathcal{S}_4}(\xi,\eta) &= ( (343/192) (49/24) r_3(\xi,\eta)r_4(\xi,\eta) )^2   
\label{eq:Appendix_S4_psi_2}
\end{align}

\bibliographystyle{spmpsci}      
\bibliography{references}

\end{document}